\documentclass{svjour3a}

\usepackage[english]{babel}
\usepackage{inputenc}
\usepackage{algorithm}
\usepackage[noend]{algpseudocode}
\usepackage{amsmath}
\usepackage{ntheorem}
\usepackage{caption}
\usepackage{comment}
\usepackage{dsfont}
\usepackage{enumitem}
\usepackage[misc]{ifsym}
\usepackage{pgf}
\usepackage{booktabs}
\usepackage{hhline}
\usepackage{boldline,multirow}
\usepackage{tikz}
\usepackage{graphicx, subfigure}

\usepackage{array, multirow}
\usetikzlibrary{arrows}
\usetikzlibrary{arrows.meta,bending}
\usetikzlibrary{decorations.pathreplacing}

\newcommand{\STAB}[1]{\begin{tabular}{@{}c@{}}#1\end{tabular}}

\usepackage{geometry}
\usepackage{natbib}
\setcitestyle{authoryear,aysep={},open={(},close={)}}

\DeclareCaptionLabelSeparator{none}{ }
\DeclareMathOperator*{\argmin}{arg\,min}

\newcommand{\red}[1]{{\color{black}#1}}

\newtheorem{defi}[theorem]{Definition}
\newtheorem{obs}[theorem]{Observation}
\newtheorem{thm}[theorem]{Theorem}
\newtheorem*{ex}{Example}

\newcommand{\M}{\mathcal M}
\newcommand{\ssq}{\sigma_0}
\newcommand{\R}{\mathcal R}
\renewcommand{\S}{\mathcal S}

\def\totwecompletion{$1|\mbox{\it resch-LIFO}|\sum w_j C_j$}
\def\maxlateness{$1|\mbox{\it resch-LIFO}|L_{\max}$}
\def\numberlatejobs{$1|\mbox{\it resch-LIFO}|\sum U_j$}
\def\wnumberlatejobs{$1|\mbox{\it resch-LIFO}|\sum w_j U_j$}
\def\maxreg{$1|\mbox{\it resch-LIFO}|\Phi_{\max}$}
\def\real{\mathds{R}}

\title{Optimally rescheduling jobs with a LIFO buffer}

\author{
Gaia Nicosia \and Andrea Pacifici \and Ulrich Pferschy \and Julia Resch \and Giovanni Righini
}

\institute{G. Nicosia 
\at Dipartimento di Ingegneria, Universit\`{a} Roma Tre, Via della Vasca Navale 79, 00146 Rome, Italy\\
\email{gaia.nicosia@uniroma3.it}
\and
A. Pacifici
\at Dipartimento di Ingegneria Civile e Ingegneria Informatica, Universit\`{a} di Roma “Tor Vergata”, Via del Politecnico 1, 00133 Rome, Italy\\
\email{andrea.pacifici@uniroma2.it}
\and
U. Pferschy \and J. Resch
\at Department of Operations and Information Systems, University of Graz, Universitaetsstrasse 15, 8010 Graz, Austria\\
\email{\{ulrich.pferschy, julia.resch\}@uni-graz.at}
\and
G. Righini
\at Dipartimento di Informatica, Universit\`{a} degli Studi di Milano, Via Celoria 18, 20100 Milan, Italy\\
\email{giovanni.righini@unimi.it}
}

\usepackage{letltxmacro}

\LetLtxMacro\OriginalLongrightarrow\Longrightarrow
\LetLtxMacro\OriginalLongleftarrow\Longleftarrow

\usepackage{trimclip}
\DeclareRobustCommand\Longrightarrow{\NewRelbar\joinrel\Rightarrow}
\DeclareRobustCommand\Longleftarrow{\Leftarrow\joinrel\NewRelbar}

\makeatletter
\DeclareRobustCommand\NewRelbar{%
	\mathrel{%
		\mathpalette\@NewRelbar{}%
	}%
}
\newcommand*\@NewRelbar[2]{%
	\sbox0{$#1=$}%
	\sbox2{$#1\Rightarrow\m@th$}%
	\sbox4{$#1\Leftarrow\m@th$}%
	\clipbox{0pt 0pt \dimexpr(\wd2-.6\wd0) 0pt}{\copy2}%
	\kern-.2\wd0 %
	\clipbox{\dimexpr(\wd4-.6\wd0) 0pt 0pt 0pt}{\copy4}%
}
\makeatother

\begin{document}

\maketitle


\begin{abstract}
This paper considers single-machine scheduling problems in which a given solution, i.e.\ an ordered set of jobs, has to be improved as much as possible by re-sequencing the jobs. The need for rescheduling may arise in different contexts, e.g.\ due to changes in the job data or because of the local objective in a stage of a supply chain \red{that is} not aligned with the given sequence.
A common production setting entails the movement of jobs (or parts) on a conveyor. 
This is reflected in our model by facilitating the re-sequencing of jobs via a buffer of limited capacity accessible by a LIFO policy.
We consider the classical objective functions of total weighted completion time, maximum lateness and (weighted) number of late jobs
and study their complexity. 
For three of these problems we present strictly polynomial-time dynamic programming algorithms,
while for the case of minimizing the weighted number of late jobs NP-hardness is proven and a pseudo-polynomial algorithm is given.

%

\end{abstract}

\section{Introduction}\label{sec:intro}


Classical single-machine scheduling problems aim at finding an optimal sequence to process a set of jobs with given processing times possibly subject to additional constraints concerning, for instance, release dates, due dates, etc.

In several industrial settings, different unforeseen phenomena, such as data-obsolescence or disruptions, could deteriorate the performance (or optimality) of the planned ahead schedule. In this case, it is sometimes possible---or even necessary---to compute a new schedule by rearranging the previous job sequence.
A similar situation frequently occurs, e.g.\ in lot production or in operating-rooms scheduling. In the first case, lots must typically go through several working stages: Between two of them, it may be beneficial to reorganize the sequence, owing to, for instance, different characteristics of the lots in the next stage. In the second case, a tentative schedule for a certain planning period is built in advance and, later on, the final schedule is output trying to minimize changes with respect to the original plan.
In this context, we are interested in the following problem: we are provided with an initial job sequence to feed a single processing resource;
we need to rearrange the jobs such that the new sequence performs well in terms of some given criterion. Depending on the considered setting, we also need to deal with a given set of feasible reconfigurations (such restrictions may be imposed, e.g.\ by the physical handling system of the plant) which are, somehow, not too distant from the original sequence. 
%
%
A special version of this problem is also considered in \citet{bib:anppAIRO2018,bib:anppCTW2018}. The authors studied a rescheduling problem 
with the constraint that the jobs extracted from the given initial sequence can be re-inserted only in later positions, i.e.\ jobs can be postponed but not moved ahead of the schedule.
This corresponds to a physical setting where jobs are transported on a conveyor that feeds a single processor and a robot or worker can pick them up and reinsert them later in the queue. 

In this paper, we adopt a similar setting but with an additional set of restrictions. In particular, we consider the scenario in which the handling mechanism consists \red{of} a conveyor that feeds a single machine \red{in a given sequence} and a robot, placed along the line, \red{that is able to alter this sequence}. 
The robot \red{may} pick parts from the conveyor as they are moving, stacks them on a buffer of finite capacity from which it takes the parts and places them back on the conveyor, in their final positions \red{(which is later in the original sequence due to the conveyor movement)}.
Since the stack is managed according to a Last-In-First-Out (LIFO) policy, only the last \red{part} put into the stack can be extracted and re-inserted in the new sequence. 
 This setting 
has been introduced in
\cite{bib:npppr2019} where the authors present some preliminary results on the corresponding rescheduling problem.
%

An area of research which is strictly related to the problem we address 
in this article, deals with {\em sequence coordination} in supply chain (SC). 
One of the main tasks in SC management is indeed coordination of several 
activities performed at different stages of the chain.
An obvious overall goal consists in successfully meeting customers
needs and achieving a good level of efficiency and performance.
Usually, in a coordinated SC, two or more processes 
subject to their mutual coordination are considered. This involves fitting 
the schedules of different manufacturers together when some planned or 
unexpected schedule changes are experienced by one or more of them \citep{bib:is2015}. 
In this context, \cite{bib:ahp2006} consider 
two consecutive stages of a supply chain where ideal job sequences 
(typically, different for the two stages) are given. 
The authors address a supply chain scheduling coordination problem consisting in 
finding a trade-off schedule that takes into account the ideal schedules
of both stages. They propose a number of polynomial-time algorithms for 
different versions of the problem, namely from the point of view of the 
manufacturer, then from that of the supplier, 
and, finally, they consider the situation when both stages cooperate to 
obtain a satisfactory compromise schedule.
A related coordination problem is addressed in \cite{bib:admp2001} where
two departments in a manufacturing facility process the jobs in batches.
In a first department a setup is paid whenever a certain attribute changes from
one batch to another, whereas in a second department, the setup 
is associated to a different attribute.
The problem of finding a unique sequence of jobs in order to minimize the overall 
setup cost arises. The authors prove that the problem is NP-hard and
propose an effective heuristic approach.
\red{
The above coordination problems are also tightly connected to those addressed in \emph{multi-agent scheduling}, 
a research field which received great attention more recently:
two or more agents have to agree on a fair, i.e acceptable schedule of their distinct
sets of jobs on a common processing resource (see, for instance, \cite{bib:LP10,bib:perez,bib:acnp2019}).
}
Another important and fruitful research stream, connected to the problem addressed here, concerns
the so-called {\em rescheduling} (or dynamic scheduling). 
In many real-world scenarios scheduling is an activity requiring frequent revisions due to unexpected changes such as, for instance, machine breakdown or unavailability, delay in the arrival of materials, job cancellation, due date changes, etc. With the terms rescheduling and dynamic scheduling many authors indicate the problem of scheduling in the presence of real-time events; 
this includes the process of updating the current schedule to face previously unknown events such as the arrival of new jobs \citep{bib:hlp2007}, disruptions \citep{bib:nbjan2018}, perturbation of the originally given or estimated data \citep{bib:hp2010}, etc. 
A recent and effective application of these concepts in the health care sector can be found in \cite{bib:bps2019}.
Two different reviews of the state-of-the-art of currently developing research on dynamic scheduling are given in \cite{bib:op2009} and \cite{bib:vhl2003}. 

The most common strategies (called predictive-reactive), when facing any unexpected change in the scenario, consider both the possibility of local adjustments and a whole re-computation of the current schedule with the aim of (locally) improving shop-efficiency.
Together with the latter objective, it is also of interest to measure how much the new schedule deviates from the original schedule. 
This concept (usually referred to as {\em stability}) is important since, typically, modification costs increase with the magnitude of such deviation, whereas big and frequent schedule changes often may cause undesired nervousness phenomena due to a lack of continuity.

From the pioneering works by \cite{bib:dk1995} and \cite{bib:wsc1993} up to the most recent papers (see, e.g.~\citealt{bib:dnpz2019,bib:nsdzc2019}), 
{\em robustness} is a widely adopted concept in scheduling and it is an alternative (pro-active) approach that tries to design a schedule which {\em a priori} guarantees a certain level of efficiency, given a set of possible scenarios. 
This way stability is preserved while performance is kept above a fixed level.

\red{
Indeed, our problem can be viewed in the framework of the so called  \emph{recoverable robustness}, see \cite{bib:LM2009}.
A recoverable robust solution is not necessarily feasible in all scenarios of a robust optimization problem but
it can be made feasible by applying a (simple, quick) recovery algorithm to it. This concept has been 
investigated in several application contexts (mostly in  transportation problems) and 
there is a limited literature also in scheduling. For instance, in \cite{bib:Akker2018}
an initial solution of a scheduling problem is given and 
in each scenario 
some {recovery actions} are performed to make the solution acceptable again. 

The algorithms presented in this work can be regarded as recovery algorithms to achieve certain
objective benchmarks (rather than feasibility) in different problem settings which are illustrated below.
}

In the special rescheduling problem addressed in this paper, we consider several objective functions, namely: total weighted completion time, maximum lateness, number of late jobs and weighted number of late jobs. 
We devise strictly polynomial-time optimization algorithms based on dynamic programming recursions for the first three objectives. 
In contrast, we prove that the problem is NP-hard when minimizing the weighted number of late jobs, but still permits a pseudo-polynomial dynamic programming algorithm in that case.
%
The remainder of this work is organized as follows: 
After a rigorous statement of the problem in Section~\ref{sec:model}, 
\red{Sections~\ref{sec:weightedcompletion}, \ref{sec:lateness}, and~\ref{sect:latejobs} 
are devoted to the first three objective functions and 
describe the corresponding {\em efficient} solution algorithms.}
For the problem with the weighted number of late jobs objective, 
in Section~\red{\ref{sec:weightedlatejobs}}, we prove its complexity and propose a pseudo-polynomial algorithm.
\red{
Section~\ref{sec:emp} illustrates by a short experimental study the behaviour of the rescheduling process depending on the stack size.
Finally, concluding remarks are given in Section~\ref{sec:conc}.
}

\section{Problem statement} \label{sec:model}

In this section, we give a formal statement of the problem under 
consideration and introduce the notation used throughout the paper. 
\red{Hereafter, we use the term ``job" to refer to both, a physical 
piece of material which is processed by some machine or resource and the 
process itself (characterized by a certain duration and possible 
additional data).}

Let us consider a deterministic single-machine environment where we are given a set $J$ of $n$ jobs that have to be scheduled according to a  {\em regular}, i.e.\ non-decreasing, objective function $f(\sigma)$
of the job completion times $C_j(\sigma)$, $j = 1,\ldots,n$. 
For such a schedule $\sigma = \langle \sigma_1, \sigma_2, \dots, \sigma_n \rangle$ with $\sigma_k\in J$, $k=1,\dots, n$, if $i \leq j$,
we refer to the ordered set of jobs $\langle\sigma_i,\sigma_{i+1},\ldots,\sigma_j\rangle$ as the {\em subsequence} $\sigma(i,j)$.
%
Moreover, for each job $j \in J$ we know its processing time $p_j$ and, possibly, a due date $d_j$ and a weight $w_j$. 
\red{
As usual for scheduling problems, we will assume that all these values are nonnegative integers.
However, this property will be required only for the dynamic programming algorithm in Section~\ref{sec:weightedlatejobs}.}
Additionally, we are given an initial sequence $\ssq$ in which the $n$ jobs of set $J$ are numbered from $1$ to $n$. 
So, $\ssq=\langle 1,2, \ldots, n \rangle$ and we say that {\em job $j$ is placed in the $j$-th position} to indicate that it is the $j$-th job of the sequence $\ssq$.
Clearly, if $i,j \in J$ and $i\le j$, we have 
$\ssq(i,j) = \langle i,i+1,\ldots,j \rangle$.

In the problem addressed here, we look for a new job sequence $\sigma$ such that 
\begin{enumerate}[label=$(\roman*)$]
    \item $f(\sigma)$ is minimum and
    \item $\sigma$ can be derived from $\ssq$ by applying a 
(constrained) number of {\em feasible moves}.
\end{enumerate}
Any move in this scheduling environment is performed by a 
physical device (e.g.\ a robot arm) that operates on
a sequence of parts, each associated to one job, arranged 
in an ordered sequence along a line (e.g.\ on a moving conveyor). 
The initial sequence on this line corresponds to $\ssq$.
The considered device 
$(i)$ picks up a job $j$;
$(ii)$ places {$j$} in the {\em stack} with bounded capacity $S$; 
$(iii)$ possibly picks up other \red{jobs} and places them in the stack;
$(iv)$ picks the \red{jobs} from the stack, according to a LIFO policy, 
and places them back (on the conveyor) at a suitable position
{\em later} in the sequence. 
We assume that there is always enough space between jobs to place even the whole content of the stack between any two \red{jobs} on the line.

In this setting, moved jobs can only be postponed, as in fact the robot arm picks the \red{jobs} from the line while the conveyor feeding the processor moves ahead.
Hence, reinsertion can only happen in a later position of the sequence.

\begin{defi}\label{defi:move}
A move $i\to j$, $i<j$, consists of deleting job $i$ from a given sequence 
and reinserting it immediately after all jobs of the subsequence $\ssq(i+1,j)$.
\end{defi}
Figure~\ref{fig:move} illustrates the process for three consecutive moves
on a sequence $\ssq$ with $n=9$ jobs. 
The picture emphasizes the characteristics of two types of feasible moves for the above mentioned physical device. 
A move $1\to 3$ in which job 1 is placed immediately after job 3 is performed first. 
Then, moves $5\to 9$ and $7\to 9$ are done:
\red{Job} $5$ is loaded into the stack directly before $7$, which is then extracted from the stack and placed just after \red{job} $9$. 
After that, $5$ is placed back on the line, right after $9$ and $7$. 
It is clear that the latter operations require that the stack capacity $S$ is
at least $2$. 
\begin{figure}[htbp]
    \centering
    \includegraphics[width=.7\linewidth]{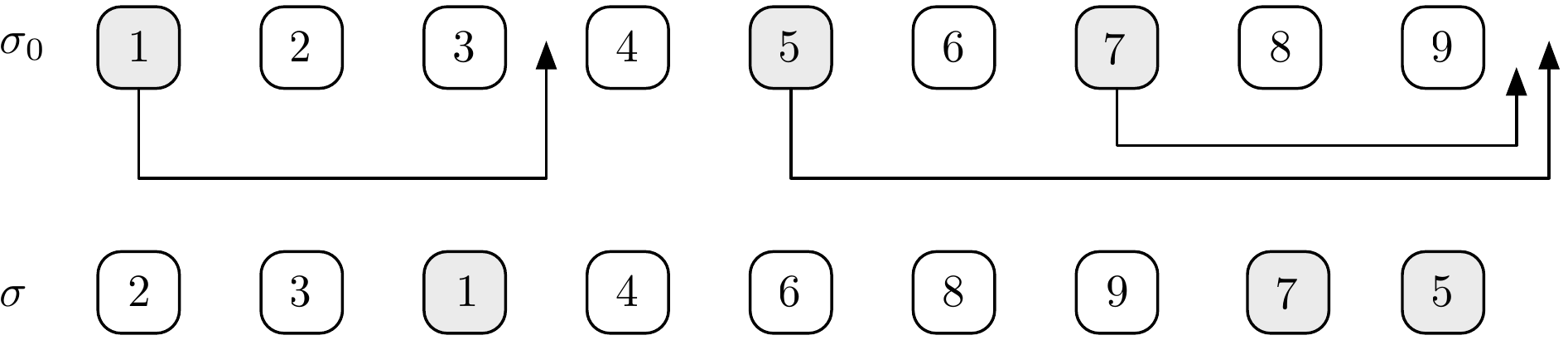}
    \caption{Three feasible moves.}
    \label{fig:move}
\end{figure}

Observe that, if we start from $\ssq$, due to the stack LIFO policy, Definition~\ref{defi:move} is consistent even after a number of moves has been performed.

Following the above considerations, we are now characterizing the compatibility of two moves.
\begin{defi}
Two moves $i_1\to j_1$ and $i_2\to j_2$ with $i_1<i_2$ are {feasible} if: 
\begin{enumerate}[label=$(\roman*)$]
    \item either $j_1 < i_2$, i.e.\ $i_1<j_1<i_2<j_2$, and the moves are called {\em sequential};
    \item or $i_1 < i_2 < j_2 \leq j_1$, and move $i_2\to j_2$ is {\em nested in} move $i_1\to j_1$.
\end{enumerate}
\end{defi}
Note that, as a consequence of Definition~\ref{defi:move}, it is easy to see that
\begin{itemize}
    \item when move $i_2\to j_2$ is nested in move $i_1\to j_1$, $i_2$ precedes $i_1$ in the final sequence even if $j_1=j_2$;
    \item although the device picks jobs in increasing index order, the sequence $\sigma$, obtained after a set of feasible moves considering the LIFO policy of the stack, is independent from the particular order in which the moves are performed. 
    \item given a sequence $\sigma \neq \ssq$, if there exists a set of feasible moves to reach $\sigma$ starting from $\ssq$, then such a set is unique.
\end{itemize}

We next define the \textit{level} of a move:
\begin{defi}\label{defi:level}
A move that does not contain nested moves is said to be at level $1$. 
Recursively, a 
move $m$ is at level $\ell > 1$ if $\ell$ is the smallest value such that $m$ contains feasible nested 
moves at level up to $\ell-1$.
\end{defi}
Clearly, a move at level $\ell$ is feasible only if $\ell \le S$.
In the example shown in Figure~\ref{fig:move}, $1 \to 3$ and $5 \to 9$ are sequential moves while $7 \to 9$ is nested in $5 \to 9$. In particular, $5 \to 9$ is at level 2.

In the remainder of this paper, we adopt the following definition for a move.
\begin{defi}
A move $i\to j$ at level $\ell$ is denoted as $(i,j,\ell)$. 
For notational convenience, we also define $(i, i, \ell)$ for every level $\ell$, meaning that job $i$ is not moved at all.
\end{defi}

Note that in our setting the stack capacity constraint imposes the index $\ell$ to be an integer within the range $1$ to $S$, 
i.e.\ there must not be more than $S$ moves nested inside each other.

\begin{defi}\label{defi:feasibleschedules}
The set of feasible schedules $\mathcal{F}_S$ for the rescheduling problem where the stack capacity is limited by $S$, 
is comprised of all the schedules resulting from a set of feasible moves at levels up to a maximum of $S$ starting from $\ssq$.
\end{defi}

In this paper we consider the minimization of three classical objective functions in scheduling theory to evaluate the sequence obtained from $\ssq$ through the LIFO-constrained moves: 
total weighted completion time, maximum lateness (with extension to any regular function), and number of late jobs. For the latter we consider both, the weighted and unweighted case. 

For any feasible schedule $\sigma$ of the jobs of $J$ and any job $j \in J$, we define the following quantities: $C_j\left(\sigma\right)$ is the completion time of $j$ in $\sigma$, $L_j\left(\sigma\right) = C_j\left(\sigma\right) - d_j$ is the lateness of job $j$ in $\sigma$ and 
$$
U_j\left(\sigma\right) = \left\{
\begin{array}{cl}
1, & \mbox{ if } C_j\left(\sigma\right) > d_j,\\
0, & \mbox{ otherwise,}
\end{array}\right.
$$ 
indicates if job $j$ is late. 
Note that the completion time of job $j$ in the initial sequence $\ssq$ is given by $C_j(\ssq) = \sum_{k=1}^j p_k$ and its lateness by $L_j(\ssq)$. We indicate with $P(i,j) = \sum_{k = i}^{j} p_k$ the total processing time of the subsequence $\ssq(i,j)$
and with $W(i,j) = \sum_{k=i}^j w_k$ its total weight. 
Moreover, for a given subsequence $\sigma(i,j)$, the maximum lateness within the subsequence is 
$$L_{\max}(\sigma(i,j)) = \max_{k = \sigma_i,\ldots,\sigma_j} \left\{ L_k(\sigma) \right\}.$$ 
Furthermore, if $\phi_j: \real_{\ge 0} \to \real$, $j \in J$, are given regular functions,
we denote the maximum of regular function objective as
$$\Phi_{\max}(\sigma) = \max_{j\in J}\{\phi_j(C_j(\sigma))\}.$$

In accordance with the standard Graham's three field notation, we indicate the problems addressed in this paper as follows:  \red{\totwecompletion, \maxlateness, \maxreg, \numberlatejobs{} and \wnumberlatejobs{}}
where the last field refers to the particular objective we want to minimize in 
rearranging the given initial schedule $\ssq$ through the above described mechanism.

\medskip
In this paper we also consider the following special generalization of the above problems.
Let $\Omega \subseteq J$ be a given subset of {\em movable} jobs and consider as feasible solutions only the schedules obtained from $\ssq$ by a set of feasible moves (at levels up to a maximum of $S$),
where only moves $i \to j$ with $i \in \Omega$ are involved.
Jobs from $J\setminus \Omega$ can not be moved, whereas 
jobs from $\Omega$ {\em may} be moved but it is not mandatory to move them.
In the remainder of the paper we will refer to this variant as {\em $\Omega$-constrained} problem. 
Observe that our original rescheduling problems are
special cases of these $\Omega$-constrained problems in which $\Omega = J$. 

\section{Total weighted completion time}  \label{sec:weightedcompletion}

We start our study by considering the minimization of the weighted sum of job completion times, for which we present a polynomial time dynamic program.

First, we can observe that for any partition of a schedule
$\sigma = \langle \sigma_1, \sigma_2, \ldots ,\sigma_n \rangle$ 
in $Q$ adjacent subsequences 
$\sigma(u_q,v_q) = \langle \sigma_{u_q},\ldots,\sigma_{v_q} \rangle$ 
with $u_q \leq v_q$
such that $u_1=1$, $v_Q=n$ and $u_q = v_{q-1} + 1$ for all $q=2,\ldots,Q$,
the objective function can be expressed as follows:
\begin{equation}  \label{def:z1}
f^{(1)}\left(\sigma\right) = \sum_{j \in J} w_j C_j\left(\sigma\right) = \sum_{q=1}^Q f^{(1)}(\sigma(u_q,v_q)), 
\end{equation}
where the term $f^{(1)}\left(\sigma(u_q,v_q)\right)$ denotes the total weighted completion time for the subsequence $\sigma(u_q,v_q)$, which in turn is the sum of the weighted completion times of the jobs in the subsequence. Denoting by $J_q$ the set of jobs of subsequence $\sigma(u_q,v_q)$, clearly we have:
\begin{equation*}
f^{(1)}\left(\sigma(u_q,v_q)\right) = \sum_{j \in J_q} w_j C_j\left(\sigma\right)\ \ \ \forall\, q=1,\ldots,Q.
\end{equation*}

Consider the set $\mathcal{M}$ of all feasible moves $(i,j,\ell)$ that allows to reach $\sigma$ starting from $\ssq$. 
Let $\mathcal{\widetilde{M}}\subseteq \mathcal{M}$ indicate the set of moves $(i,j,\ell)$ which are not nested in any other move $(h,k,\ell')\in\mathcal{M}$ at a higher level $\ell'>\ell$. 
By definition, the moves $\mathcal{\widetilde{M}}$ are sequential moves.
In $\mathcal{\widetilde{M}}$ we also include {\em identity} moves $(i,i,\red{1})$ which correspond to not moving job $i$. (Clearly, parts of $\sigma$ which are unchanged with respect to $\ssq$, i.e.\ if the subsequence 
$\sigma(i,j)=\ssq(i,j)$, can be represented by a set of consecutive identity moves $(k,k,\red{1})$ with $k=i, \dots, j$.)
Let us denote $\mathcal{\widetilde{M}}$ as follows:
\begin{align}  \label{eq:decomp}
\widetilde{\M} = \left\{ (u_1,v_1, \ell_1), (u_2,v_2,\ell_2), \ldots, (u_r,v_r,\ell_r)\right\}
\end{align} 
with $1= u_1 \le v_1$, $v_r = n$, $v_{q-1} + 1 = u_{q} \le v_q$, $q=2,\ldots, r$. 
Because of the additivity principle \eqref{def:z1}, the cost $f^{(1)}(\sigma)$ of schedule $\sigma$ can be computed by adding---to the cost of $\ssq$---the contributions of the $r$ feasible moves of $\widetilde{\M}$.
%
Note that in \eqref{eq:decomp} each move $(u_q,v_q,\ell_q)$ may contain nested moves, so that its contribution depends on these nested moves as well.
%
For instance, the sequential moves of $\widetilde{\M}$ in the example illustrated by Figure~\ref{fig:move} are $(1,3,1)$, $(4,4,1)$ and $(5,9,2)$. 
Move $(5,9,2)$ contains the nested move $(7,9,1)$ and therefore its contribution must also account for the one of $(7,9,1)$.


Hereafter, we present a dynamic programming algorithm for determining the optimal set of moves yielding a minimum cost schedule $\sigma^*$, starting from $\ssq$. 
The correctness of this algorithm straightforwardly follows from the following observation
which corresponds to a standard optimality principle: 
In any optimal schedule $\sigma^*$, there exists a partition as in \eqref{eq:decomp}, 
such that each subsequence in the partition is optimal for the subproblem containing 
only the jobs of that subsequence.
The basic step of the dynamic program consists of computing the cost of a move $(i,j,\ell)$, 
at level $\ell$ based on the knowledge of optimal costs for any subsequence of $\sigma(i,j)$ in which the stack capacity is $\ell - 1$. 
\begin{defi}
The cost $c(i,j,\ell)$ is defined as the {\em minimum} total weighted completion time {\em variation} when a move $(i,j,\ell')$ at some level $\ell' \leq \ell$ is performed.
The cost $\mu^*(i,j,\ell)$ is the \red{minimum cost variation of }
subsequence $\ssq(i,j)$, when the stack capacity is equal to $\ell$. 
\end{defi}
Observe that $c(i,j,\ell)$ includes the \red{optimal cost variation produced by} 
all possible nested moves at lower levels and hence it takes into account information of all $c(h,k,\ell-1)$ for $i<h\le k\le j$ which in turn are summed up into $\mu^*(i+1,j,\ell-1)$.
While $\mu^*(i,j,\ell)$ gives the \red{optimal cost variation} 
for the subproblem implied by $i$ and $j$, 
$c(i,j,\ell)$ considers the particular solution where the move $i\to j$ must be performed.

Hereafter, we show how to compute the variations in the schedule overall cost, for different moves and their combinations.
%
%
We first consider the effect of {\em single moves} on the total weighted completion time.
In this context, the following statement results directly from the definition of the completion time.
\begin{obs}\label{prop:TWCT_noEffects}
The completion time of each job preceding $i$ and following $j$ in the initial sequence $\ssq$ is not affected by the move $(i,j,\ell)$.
\end{obs}

Based on the initial schedule $\ssq$, the {\em cost variation} $m(i,j)$ due to moving job $i$ after job $j$ is expressed as
\begin{equation}
m(i,j) = w_i \sum_{k=i+1}^{j} p_k - p_i \sum_{k=i+1}^j w_k \ \ \ \forall\, i,j \in J, \ i<j. \label{def_m}
\end{equation}
The first term in (\ref{def_m}) indicates the increase of the weighted completion time of job $i$ while 
the second term indicates the total decrease of the weighted completion times of the jobs in $\ssq(i+1,j)$. 
In addition, because of the definition of the cost of a move~\eqref{def_m}, it is easy to verify the following:
\begin{obs}\label{prop:effectMove1}
The effect of a move $(i,j,\ell)$ at any level $\ell>1$ does not depend on the order of the jobs in the subsequence $\ssq\left(i+1,j\right)$.
\end{obs}
As a consequence, because of the LIFO constraint, the cost variation due to a {\em single} move $(i,j,\ell)$ is always equal to $m(i,j)$, for all possible levels $\ell$.

As far as {\em sequential moves} are concerned, we recall that, 
owing to Equation~\eqref{def:z1}, the effect of two sequential moves $(i_1,j_1,\ell_1)$ and $(i_2,j_2,\ell_2)$, $j_1 < i_2$, on the objective $f^{(1)}$ equals the sum of the effects of each move.
%



Moves at level $\ell = 1$ can not contain any nested moves, so the cost $c(i,j,1)$ of a move is equal to $m(i,j)$. Obviously, when a job $i$ is not moved, the associated cost is null and $c(i,i,1)=0$.
Hence, we may give a recursive expression of the optimal cost for 
the subsequence $\ssq(i,j)$ 
as follows:
\begin{align}\label{eq:mu_at_level1}
\left\{
    \begin{array}{ll}
     \mu^*(i,i,1) = 0 &\ \ \forall\, i \in J\\
     \mu^*(i,j,1) =\min\{ \min_{k \in \ssq(i,j-1)} \{c(i,k,1) + \mu^*(k+1,j,1)\}, c(i,j,1)\} &\ \ \forall\, i, j\in J,\ i < j.
    \end{array}
\right.
\end{align}

Even when dealing with {\em nested moves}, i.e.\ for $\ell>1$, the 
cost variation of a set of moves is additive: Following the same line of 
arguments yielding to Observation~\ref{prop:effectMove1} and Equation~\eqref{def_m}, 
the overall effect on $f^{(1)}$ of two (or more) nested moves is equal 
to the sum of the effects of each move. 

In order to compute the cost $c(i,j,\ell)$ for a move $(i,j,\ell')$ with $\ell' \leq \ell$, 
one clearly has to take into account the possibility of optimally rearranging the subsequence $\ssq(i+1,j)$ with moves at levels lower than $\ell'$. 
Thus, after computing all optimal costs $\mu^*(i,j,\ell-1)$ obtainable from a set of moves up to level $\ell-1$ for the subsequence $\ssq(i,j)$, it is possible to determine the (optimal) cost $c(i,j,\ell)$.
This in turn allows to compute the optimal costs $\mu^*(i,j,\ell)$ for the subsequence $\ssq(i,j)$ at level $\ell$. Formally, the following recursion holds:
\begin{align}\label{eq:mu_recursion}
\left\{
\begin{array}{ll}
c(i,j,1) = m(i,j)                          &\ \ \  \forall\, i,j \in J,\ i< j \\
c(i,i,\ell) = 0                            &\ \ \  \forall\, i \in J,\ \forall\, \ell=1,\ldots,S \\
c(i,j,\ell) = m(i,j) + \mu^*(i+1,j,\ell-1) &\ \ \  \forall\, i,j \in J,\ i<j,\ \forall\, \ell=2,\ldots,S.
\end{array}
\right.
\end{align}
In Equation~\eqref{eq:mu_recursion}, $\mu^*(i,j,\ell)$ is computed in a similar way as in~\eqref{eq:mu_at_level1}:
\begin{align}  \label{eq:mu_at_levelell}
\left\{
    \begin{array}{ll}
     \mu^*(i,i,\ell) = 0                                                                         &\ \ \forall\, i \in J \\
     \mu^*(i,j,\ell) =\min\{ \min_{k \in \ssq(i,j-1)} \{c(i,k,\ell) + \mu^*(k+1,j,\ell)\}, c(i,j,\ell)\} &\ \ \forall\, i, j\in J,\ i < j. 
    \end{array}
\right.
\end{align}

Finally, the optimal solution value is found by computing $\mu^*(1,n,S)$ such that the objective function value of the optimal schedule equals
$$f^{(1)}(\sigma^*) = f^{(1)}(\ssq) + \mu^*(1,n,S).$$

As for the computational complexity of the above dynamic program, we observe that
the computation of all $m(i,j)$, i.e.\ the costs at level $\ell=1$, requires $O(n^2)$ time. 
For each level $\ell= 1,\ldots,S$ and for each ordered pair of jobs $(i,j)$, the algorithm must compute the cost $c(i,j,\ell)$ of the corresponding move and the optimal subsequence cost $\mu^*(i,j,\ell)$. 
Computing all $c$ values requires $O(n^2)$ time at each level, since each value is computed in $O(1)$, while the computation of $\mu^*(i,j,\ell)$ in Equation~\eqref{eq:mu_at_levelell} has a cost of $O(n^3)$ time for each level $\ell$. Hence, the overall complexity of the algorithm is $O(n^3 S)$ which is upper bounded by $O(n^4)$.

The results above can be summarized as follows.
\begin{thm}
Problem \totwecompletion{} is polynomially solvable in time $O(n^4)$.
\end{thm}

\subsection{Extension to the $\Omega$-constrained problem}  \label{sec:omega_wc}

It is easy to handle the special case in which only jobs from a given subset $\Omega \subseteq J$ of the jobs can be moved. 
For this purpose, we just make moves of jobs in $\Omega$ extremely costly. 
Hence, we can \red{still} use Recursions~\eqref{eq:mu_recursion} and \eqref{eq:mu_at_levelell} \red{but}
with the following adjusted \red{costs at level $\ell=1$} for a large enough constant $M$:
\red{
$$
c(i,j,1) = 
\left\{
\begin{array}{ll}
M, &  \mbox{ if }  i\in J\setminus\Omega, \ i < j,\\
\mbox{as in Expression~\eqref{eq:mu_recursion}}, &  \mbox{ otherwise},
\end{array}\right.
$$
for all $i,j \in J$ with $i \leq j$.}

\section{Maximum of regular functions minimization}  \label{sec:lateness}

In this section, we first focus on the minimization of the maximum lateness of the jobs and
present a dynamic programming solution algorithm related to the one presented in the previous Section~\ref{sec:weightedcompletion}. 
We then discuss how the same algorithm can be used to solve the more general problem \maxreg. 


\subsection{Minimization of maximum lateness}

Recall that for a given schedule $\sigma$, we indicate by $L_j(\sigma) = C_j(\sigma) - d_j$, $L_{\max}(\sigma) = \max_{j\in J}\{L_j(\sigma)\}$, and $L_{\max}(\sigma(i,j)) = \max_{k=\sigma_i,\dots,\sigma_j}\{L_k(\sigma)\}$.

Here, the objective function for any partition of a schedule $\sigma = \langle \sigma_1, \sigma_2, \ldots ,\sigma_n \rangle$ in $Q$ adjacent subsequences $\sigma(u_q,v_q) = \langle \sigma_{u_q},\ldots,\sigma_{v_q} \rangle$ with $u_q \leq v_q$ such that $u_1=1$, $v_Q=n$ and $u_q = v_{q-1} + 1$ for all $q=2,\ldots,Q$, fulfills the following decomposition property:
\begin{equation*}
L_{\max}(\sigma) = \max_{q = 1,\ldots,Q} \left\{L_{\max}(\sigma(u_q,v_q))\right\}.
\end{equation*}

We first consider the effect of single moves on the maximum lateness. A straightforward consequence of Observation~\ref{prop:TWCT_noEffects} is that a move $(i,j,\ell)$ does not affect the lateness of any job preceding $i$ and following $j$ in the initial sequence $\ssq$.
Hence, even in this case, moves involving disjoint subsequences can be evaluated independently.

Differently from Equation~\eqref{def_m}, 
we are not looking at the {\em variation} in the objective function due to a move $(i,j,\ell)$, but rather at the optimal value that the objective function would take if move $(i,j,\ell)$ were executed. 
Hence, we define the cost $g(i,j,\ell)$ as the {\em minimum} value of the maximum lateness in subsequence $\sigma(i,j)$ when 
\red{ moves up to level $\ell$ are done within it.}
For a single move $(i,j,\ell)$, at any level, the lateness of all jobs in $\ssq(i+1,j)$ decreases by $p_i$ while the lateness of job $i$ increases by $P(i+1,j)=\sum_{k=i+1}^j p_k$. As a consequence, we may compute the cost at level 1 as follows:
\begin{equation*}
g(i,j,1) = \max \left\{L_i(\ssq) + P(i+1,j),\ {L_{\max}}(\ssq(i+1,j)) - p_i \right\}\ \ \ \forall\, i,j \in J,\ i < j.
\end{equation*}
The term ${L_i(\ssq)} + P(i+1,j)$ indicates the new lateness of job $i$ while
the term ${L_{\max}}(\sigma(i+1,j)) - p_i$ indicates the new maximum lateness of the subsequence $\ssq(i+1,j)$. 
We also define the cost of not moving a job $i$:
\begin{equation*}
g(i,i,1) = L_i(\ssq)\ \ \ \forall\, i \in J. \label{1}
\end{equation*} 
 
Similar to the problem with total weighted completion time objective, also in this case the dynamic programming algorithm is based on computing the cost of a move $(i,j,\ell)$, at level $\ell$ starting from the optimal costs of any subsequence of $\sigma(i,j)$ at level $\ell - 1$. 
\begin{defi}
The cost $\lambda^*(i,j,\ell)$ is the value of the optimal solution of the subproblem restricted to subsequence $\ssq(i,j)$, when the stack capacity is equal to $\ell$. 
\end{defi}
Note that, as before, for $\ell>1$ the computation of $g(i,j,\ell)$ takes into account the lateness values 
\red{produced by} 
all possible nested moves at lower levels, i.e.\ the quantities $g(h,k,\ell-1)$ for $i<h\le k\le j$ which in turn yield the value $\lambda^*(i+1,j,\ell-1)$.

In general, different orders of the jobs in the subsequence $\ssq(i+1,j)$ imply different effects of a move $(i,j,\ell)$ on the objective function.
Similar to the total weighted completion time case, when computing the cost $g(i,j,\ell)$
we need to take into account the possibility of optimally rearranging 
the subsequence $\ssq(i+1,j)$ with moves at lower levels. 
Thus, only after computing the optimal cost $\lambda^*(i+1,j,\ell-1)$, 
obtainable from a set of moves at level 
at most
$\ell-1$ for the subsequence $\ssq(i+1,j)$, it is possible to determine 
$g(i,j,\ell)$.

Starting from the $g(\cdot)$ values at level $\ell$, it is then possible to compute the optimal costs $\lambda^*(i,j,\ell)$ for each subsequence $\ssq(i,j)$ at level $\ell$. 

In conclusion, the following recursion holds:
\begin{align}\label{eq:d_recursion}
\left\{
\begin{array}{ll}
g(i,j,1) = \max \{L_i(\ssq) + P(i+1,j),\ {L_{\max}}(\ssq(i+1,j)) - p_i \} & \ \ \ \forall\, i,j \in J,\ i < j\\
g(i,i,\ell) = {L_i(\ssq)} & \ \ \ \forall\, i \in J,\ \forall\,\ell=1,\ldots,S\\
g(i,j,\ell) = \max \{ L_i(\ssq) + P(i+1,j) , \lambda^*(i+1,j,\ell-1) - p_i \} & \ \ \ \forall\, i \in J,\ i < j,\ \forall\, \ell = 2,\ldots,S.
\end{array}
\right.
\end{align}
With Equation~\eqref{eq:d_recursion}, $\lambda^*(i,j,\ell)$ is computed in an analogous way as done in~\eqref{eq:mu_at_level1}:
%
%

\begin{align}\label{eq:lambda_at_levelell}
\left\{
    \begin{array}{ll}
     \lambda^*(i,i,\ell) = L_i(\ssq) & \ \ \ \forall\, i \in J, \ \forall\, \ell = 1,\ldots,S \\
     \lambda^*(i,j,\ell) = \min\{\min_{k \in \ssq(i,j-1)} \{ \max\{g(i,k,\ell) , \lambda^*(k+1,j,\ell)\}\}, g(i,j,\ell)\} & \ \ \ \forall\, i,j \in J,\ i < j,\ \forall\,\ell = 1,\ldots,S. 
    \end{array}
\right.
\end{align}

In the above Equation~\eqref{eq:lambda_at_levelell}, $\lambda^*(i,i,\ell)$ always equals the ``initial'' value ${L_i(\ssq)}$ while, only when we consider a (proper) subsequence with more than one job, the actual costs associated with moves are accounted for.
In particular, the optimal solution value $\lambda^*(i,j,\ell)$ can be computed as the best alternative among the solutions in which a move $(i,k,\ell)$, $k=i,\ldots,j$, is done: option $k=i$ corresponds to a solution in which \red{job $i$ is not moved};
if the best alternative is $k=j$, then $\lambda^*(i,j,\ell) = g(i,j,\ell)$; if $i<k<j$, the maximum lateness is the largest between the cost of move 
$i \rightarrow k$ at level at most $\ell$
and the optimal cost of the subsequence $\ssq(k+1,j)$ at level 
at most $\ell$.
Finally, the optimal solution value is found by computing $f^{(2)}(\sigma^*) = \lambda^*(1,n,S)$.

\smallskip
The computational complexity can be computed as in Section~\ref{sec:weightedcompletion}. The initialization, i.e.\ the computation of costs at level $\ell= 1$, requires $O(n^2)$ time.
For each level $\ell=1,\ldots,S$ the algorithm must compute the costs $g(\cdot)$ and $\lambda^*(\cdot)$ for each pair of jobs $(i,j)$ (with $i$ preceding $j$ in $\ssq$).
Computing all $g$ values requires $O(n^2)$ time at each level since each value is computed in $O(1)$.
The optimal costs $\lambda^*(i,j,\ell)$ can be computed, through Recursion~\eqref{eq:lambda_at_levelell}, 
in $O(n^3)$ for each level $\ell$.
Hence, the overall complexity of the algorithm is $O(n^3 S)$ which is upper bounded by $O(n^4)$.

Summarizing, we have the following result:
\begin{thm}
Problem \maxlateness{} is polynomially solvable in time $O(n^4)$.
\end{thm}

\subsection{Extension to maximum of regular functions objective}

The correctness of the above dynamic programming algorithm 
is based on a decomposition principle ensuring that parts of an
optimal schedule are also optimal for the subproblems containing only 
the jobs of those parts. This in turn implies that an optimal 
subsequence is independent of its starting time. In other words,
consider an optimal schedule $\sigma^*$ and one of its 
subsequences $\sigma^*(i,j)$ starting at, say, time $t$.
Then, the same sequence $\sigma^*(i,j)$ is optimal for the subproblem 
pertaining only jobs $\{i,i+1,\ldots,j\}$ (starting, e.g.\ at time $t=0$).

Unfortunately, when dealing with the minimization of the maximum 
of regular functions of the job completion times, such a property
does not hold anymore. In fact, consider as an example the following 
trivial (sub-)problem with only two jobs $\{1,2\}$, with $p_1 =1$,
$p_2 =4$ and linear regular functions $\phi_1(x) = 0.2x + 15$, 
$\phi_2(x) = 2x$. The optimal schedule starting, e.g.\ at time $t = 0$ 
is $\langle 1, 2\rangle$, while when the starting time is $t=10$, the optimal sequence becomes $\langle 2, 1\rangle$. As a consequence,
no immediate generalization of the efficient dynamic programming approach illustrated Recursions~\eqref{eq:d_recursion} and~\eqref{eq:lambda_at_levelell} may guarantee optimality.

However, problem \maxreg{} is still efficiently solvable by using 
a simple alternative procedure presented hereafter. 
The idea is to perform a binary search over the optimal objective values $\Phi_{\max}$, each time solving a suitable instance of \maxlateness{}.

First observe that we may compute an interval in which the optimal value
of the objective varies:
\begin{align}\label{eq:Phimaxinterval}
    \alpha = \max_{j\in J}\{\phi_j(p_j)\}\le \Phi_{\max}\le \max_{j\in J}\{\phi_j(P(1,j))\} = \omega.
\end{align}

Let $I$ be an instance of \maxreg{} and $\bar\Phi$ be a fixed target value for the optimal objective value of $I$. 
Clearly, $\bar\Phi$ must be chosen in the interval $[\alpha, \omega]$. 
We are asking if there exist solution schedules for our instance of \maxreg{} with optimal objective value not larger than $\bar\Phi$. We refer to such an instance as a YES-instance for $\bar\Phi$.

For all jobs $j\in J$, compute a deadline 
$\tau_j({\bar\Phi}) = \max_{t\ge 0}\{\phi_j(t)\le \bar\Phi\}$.
In general, since $\phi_j(\cdot)$ is non-decreasing, each such deadline may be efficiently computed in time, say, $c$. 
In the worst case, $c$ is $O(\log(\sum_{j\in J} p_j))$,
however it is a reasonable to think of special cases where $c$ is constant.
Then solve an instance of \maxlateness{} with 
due date $d_j = \tau_j({\bar\Phi})$ for $j \in J$. Let $\sigma^*(\bar\Phi)$ be an optimal solution and $L_{\max}(\bar\Phi)$ be the corresponding optimal objective value. If $L_{\max}(\bar\Phi)\le 0$, then $I$ is a YES-instance for $\bar\Phi$ and $\sigma^*(\bar\Phi)$ is a solution of \maxreg{} with objective not larger than $\bar\Phi$. 

Then an optimal solution for our instance $I$ of \maxreg{} can be found by looking for the minimum value of $\bar\Phi$ such that $L_{\max}(\bar\Phi)\le 0$. As we already mentioned, this can be done by performing a binary search over the interval $[\alpha,\omega]$ of \red{Expression}~\eqref{eq:Phimaxinterval}.

The overall complexity of the above procedure depends on the computational costs of the dynamic programming for the \maxlateness{} problem, the deadlines computation, and the binary search. In conclusion, assuming that we are able to efficiently compute the values $\phi_j(\cdot)$ and $\tau_j(\cdot)$  for all $j\in J$, and the values of suitable bounds for the objective $\alpha$ and $\omega$, we may state that:
\begin{thm}
Problem \maxreg{} is polynomially solvable in time 
$O((n\cdot c + n^4)\log(\omega - \alpha))$ 
where $c$ is the cost for computing a job deadline.
\end{thm}

\subsection{Extension to the $\Omega$-constrained problem}  \label{sec:omegalateness}

Similar to Section~\ref{sec:omega_wc}, in order to handle the special case where not 
all jobs are movable, the initialization in Recursion~\eqref{eq:d_recursion} is adjusted as
follows (for a large enough constant $M$):
\red{
$$
g(i,j,1) = 
\left\{
\begin{array}{ll}
M, &  \mbox{ if }  i\in J\setminus\Omega, \ i < j,\\
\mbox{as in Expression~\eqref{eq:d_recursion}}, &  \mbox{ otherwise},
\end{array}\right.
$$
for all $i,j \in J$ with $i \leq j$.}

\section{Minimization of the number of late jobs}\label{sect:latejobs}

This section provides a dynamic programming algorithm to minimize the number of late jobs, $f^{(3)}(\sigma) = \sum_{j \in J} U_j(\sigma)$, where the \red{quantity} $U_j(\sigma)$ assumes value $1$ if and only if job $j$ is late in schedule $\sigma$, i.e.\ if its lateness $L_j(\sigma)$ is strictly positive.

Given a subsequence of a schedule, the number of late jobs in this subsequence depends not only on the processing times and due dates of the jobs in this subsequence, but generally also on its starting time. 
\red{As a consequence, an (optimal)}
rearrangement of a subsequence  that minimizes the number of late jobs depends on \red{its} starting time.
This is illustrated in the following example:  
consider the sequence $\sigma = \langle 1,2,3 \rangle$ with $p_1 = 7$, $p_2 = p_3 = 10$, $d_1 = d_2 = 25$ and $d_3 = 15$.
It is easy to see that, in order to minimize the number of late jobs in the subsequence $\sigma(2, 3)$, 
the ordering $\langle 3,2 \rangle$ should be preferred if the sequence starts at $t=0$ (i.e.\ move $1 \to 3$ is performed), whereas for $t = 7$ (i.e.\ job $1$ is not moved) the arrangement $\langle 2,3 \rangle$ is better. 
However, the following result indicates that 
\red{the relative order of the}
lateness values are independent from the starting time of a subsequence.

\red{\begin{obs} \label{prop:qthLargest}
Let $\sigma_q(i,j)$ be a feasible arrangement of the subsequence $\sigma(i,j)$ that minimizes the $q$-th largest lateness value, $q=1,\ldots,j-i+1$. Then $\sigma_q(i,j)$ 
does not depend on the starting time of $\sigma_q(i,j)$.
\end{obs}}

Note that in general the arrangements minimizing the $q$-th largest lateness are different for different values of $q$.
However, when a subsequence is moved \red{earlier} (or delayed), all lateness values decrease (or increase) by the same amount and hence their order remains unchanged. \red{Observation~\ref{prop:qthLargest} holds for any possible definition of ``feasible arrangement''; in particular, in this paper we clearly refer to the set of sequences resulting from a set of feasible moves at levels up to a given maximum. 

In the reminder of this section, we refer to {\em $\ell$-rearrangement} to indicate a feasible rearrangement of a subsequence that can be obtained by moves up to level $\ell$.
}

\subsection{Dynamic Programming}\label{sec:dpU}

By exploiting Observation~\ref{prop:qthLargest}, we subsequently devise a strongly polynomial time dynamic program for minimizing the number of late jobs.

In the dynamic program, we store for each subsequence of jobs not only the current number of late jobs, but we also record 
the minimum required reduction of the starting time for this subsequence to reach any 
number of late jobs\red{, as we formally define below:
\begin{defi}\label{def:state}
For each $i,j \in J$ with $i \leq j$, $m=0,\ldots,j-i$ and $\ell = 0,\ldots,S$, let $s(i,j,m,\ell)$ 
be the {\em minimum decrease of the starting time} of the subsequence $\ssq(i,j)$ to obtain $m$ late jobs in it when the subsequence can be rearranged with moves up to level $\ell$.
\end{defi}
}
Hereinafter, we will refer to this information as the {\em state} of the dynamic programming algorithm.
\red{Note that $s(i,j,m,\ell)$ may also be negative (i.e.\ an increase of the starting time) if the number of late jobs in $\ssq(i,j)$ is smaller than $m$.
}
%
%
\red{Definition~\ref{def:state} does not include the case $m = j-i+1$
since the case where {\em all} jobs in a subsequence are late never requires reducing the starting time of the subsequence and thus is not relevant for a LIFO rearrangement.
}
%
%

Due to Observation~\ref{prop:qthLargest}, all $\ell$-rearrangements of a subsequence $\ssq(i,j)$ that allow us to obtain $m$ late jobs in it, when the subsequence is moved \red{earlier} by $s'$, are dominated by the $\ell$-rearrangements of $\ssq(i,j)$ that allow to obtain the same number of $m$ late jobs in it when the subsequence is \red{moved} by $s'' < s'$.
By this dominance relation between rearrangements, we can restrict ourselves to recording the minimum decrease of the starting time\footnote{\red{Clearly, the values $s(i,j,m,\ell)$ are finite, even if they are negative. Since $m<j-i+1$, there always remains at least one on-time job which forbids an arbitrarily large increase of the starting time.}}
and thus avoid the combinatorial explosion of the problem.


For every given number $m$, the value of state $s(i,j,m,\ell)$ is obtained by taking the \red{minimum} 
reduction of the starting time over all possible moves $(i,k,\ell')$
for $k=i,\ldots,j$ at some level $\ell' \leq \ell$ 
\red{that produce (at most) $m$ late jobs in $\ssq(i,j)$. (Note that if $k=i$ we are looking at an identity move, 
while if $k>i$ we are considering all possible moves nested in $i\to k$.)
}

For each candidate $k$ we \red{compute three multi-sets of lateness values, coming from three job subsets:
namely the jobs nested in the move, the moved job and the other jobs.
From the multi-set obtained by the union of these multi-sets, we select the minimum reduction in time that is needed to have $m$ late jobs in $\ssq(i,j)$ for each value of $m$.}
In particular, 
\red{some} 
late jobs \red{can} 
be obtained from subsequence $\ssq(i+1,k)$ after \red{reducing its starting time} by $p_i$;
\red{some} 
late jobs \red{can} 
be obtained from subsequence $\ssq(k+1,j)$ whose starting time has not changed;
\red{and an additional late job is possibly the moved job $i$}
\red{that is delayed} 
by $P(i+1,k)$.

At level $\ell=0$, the values $s(i,j,m,0)$ are evaluated on the initial sequence using the lateness values of the jobs in it.
For each subsequence $\ssq(i,j)$, we use the notation $\max^{(q)}_{k \in \ssq(i,j)} \{L_k(\ssq)\}$ to indicate the $q$-th largest value \red{in the multi-set $\{L_k(\ssq): k \in \ssq(i,j)\}$ that contains} 
the lateness values of the jobs in the subsequence.
To be precise, we refer to the $q$-th entry in the non-increasingly sorted \red{multi-set.} 
With this notation, we initialize
\begin{equation}  \label{equ:init_late}
s(i,j,m,0) = \max_{k \in \ssq(i,j)}{\!\!\!\!\!}^{(m+1)} \{L_k(\ssq)\} \ \ \ \forall\, i,j \in J,\ i \leq j,\ \forall\, m=0,\ldots,j-i. 
\end{equation} 

\red{For levels $\ell\geq 1$,} the recursive extension rule is as follows. 
For each subsequence $\ssq(i,j)$ 
and for each level $\ell=1,\ldots,S$, 
we first consider all possible ways in which this subsequence can be 
feasibly rearranged with moves up to level $\ell$.
\red{For this purpose, 
we compute the multi-set $\S(i,j,k,\ell)$ containing, for each possible move $(i,k,\ell')$ at any level $\ell' \leq \ell$, the $s(\cdot)$ values 
for all numbers $m = 0, \ldots, j-i$ of late jobs:   
}
\begin{equation} \label{equ:rec_multiset}
\S(i,j,k,\ell) = \bigcup_{u = 0}^{k-i-1} \{s(i+1,k,u,\ell-1)-p_i\} \cup \bigcup_{u = 0}^{j-k-1} \{s(k+1,j,u,\ell)\} \cup \{C_k(\ssq)-d_i\}.
\end{equation}
In \eqref{equ:rec_multiset}, 
\red{we distinguish the union of three multi-sets: The first one contains $k-i$ values computed from the non-dominated $(\ell-1)$-rearrangements of the subsequence $\ssq(i+1,k)$ whose jobs, due to the move $i \to k$, start earlier by $p_i$ time units. 
The second multi-set refers to the $s(\cdot)$ values of the jobs from $k+1$ to $j$ whose starting times are not affected by the move. Finally, the third term corresponds to the lateness of job $i$ after  move $i \to k$ is performed.}
Recall that Observation~\ref{prop:qthLargest} guarantees that a $(\ell-1)$-rearrangement 
of subsequence $\ssq(i+1,k)$ producing (at most) $m$ late jobs does not change when the 
subsequence is moved \red{earlier} by $p_i$.

\red{
The $(m+1)$-largest entry 
of the multi-set $\S(i,j,k,\ell)$ represents the 
necessary amount of time by which we need to bring forward $\ssq(i,j)$ after a move $(i,k,\ell')$ with $\ell' \leq \ell$,
in order to have $m$ late jobs.
Thus, the value 
of the new state corresponds to the move 
$i \to k$ such that this decrease of the starting time is minimum and
can be computed as follows:}
\begin{align} \label{equ:rec}
s(i,j,m,\ell) = \min_{k \in \ssq(i,j)} \left\{\max{\!}^{(m+1)} \S(i,j,k,\ell) \right\}.
\end{align}

The optimal solution value, $f^{(3)}(\sigma^*)$ \red{$= \sum_{j \in J} U_j(\sigma^*)$}, is finally given by $\min\{m: s(1,n,m,S) \leq 0\}$.


\red{
To illustrate the nontrivial handling of multi-sets in the dynamic programming recursion, we give the following example.}

\begin{ex}
\red{
Consider a schedule $\sigma_0 = \langle 1, 2, 3, 4 \rangle$ with $4$ jobs. The respective processing times and due dates are given in Table~\ref{tab:exInit}. 
Since the example is intended to illustrate how a state value is computed, we restrict ourselves to the case $S=1$. An analogous procedure is to be used for larger values of $S$.

In order to compute the minimum number of late jobs in the sequence $\ssq$ with moves up to level $1$, we have to calculate all state values $s(i,j,m,\ell)$ with $i,j \in \{1,2,3,4\}$, $i\leq j$, $m= 0,\ldots,j-i$ and $\ell \in \{0,1\}$. 
}
\begin{table}[htbp]
\red{
	\begin{center}
		\begin{tabular}{ccccc}
		    \toprule
			job         &  $1$ &  $2$ &  $3$ &  $4$  \\  \midrule
			$p_i$       & $25$ & $10$ &  $5$ & $10$  \\
			$d_i$       & $45$ & $15$ & $10$ & $30$  \\  \midrule
			$C_i(\ssq)$ & $25$ & $35$ & $40$ & $50$  \\
			$L_i(\ssq)$ &$-20$ & $20$ & $30$ & $20$  \\
			\bottomrule
		\end{tabular}
	\caption{\red{Processing times, due dates, completion times and lateness values of an initial schedule $\sigma_0 = \langle 1, 2, 3, 4 \rangle$ with $3$ late jobs.}}
	\label{tab:exInit}
	\end{center}
}
\end{table}

\red{
At level $\ell = 0$, the initialization according to \eqref{equ:init_late} needs to be performed. 
Since the lateness value of each job in $\ssq$ is already given in Table~\ref{tab:exInit}, the $s(\cdot)$ values can easily be obtained. For example, $s(1,4,2,0)$ corresponds to the third largest value in the multi-set $\{30, 20, 20, -20\}$ which is $20$. All state values at level $0$ are presented in Table~\ref{tab:exState0}. 
}

\begin{table}[htb]
	\red{
		\begin{center}
			\renewcommand{\arraystretch}{1.2}
			\begin{tabular}{cc|c|cccc}
				\toprule
				\multicolumn{3}{c|}{} & \multicolumn{4}{c}{$m$} \\
				\multicolumn{3}{c|}{} & $0$ & $1$ & $2$ & $3$ \\ \midrule\midrule
				$j=1$ & 
				  $i=1$ &  $s(1,1,m,0)$  & $-20$  &        &        &         \\ \midrule\midrule
				\multirow{2}{*}{$j=2$} & 
				  $i=2$ &  $s(2,2,m,0)$  &  $20$  &        &        &         \\ \cline{2-7}
				& $i=1$ &  $s(1,2,m,0)$  &  $20$  & $-20$  &        &         \\ \midrule\midrule
				\multirow{3}{*}{$j=3$} & 
				  $i=3$ &  $s(3,3,m,0)$  &  $30$  &        &        &         \\ \cline{2-7}
				& $i=2$ &  $s(2,3,m,0)$  &  $30$  &  $20$  &        &         \\ \cline{2-7}
				& $i=1$ &  $s(1,3,m,0)$  &  $30$  &  $20$  & $-20$  &         \\ \midrule\midrule
				\multirow{3}{*}{$j=4$} & 
				  $i=4$ &  $s(4,4,m,0)$  &  $20$  &        &        &         \\ \cline{2-7}
				& $i=3$ &  $s(3,4,m,0)$  &  $30$  &  $20$  &        &         \\ \cline{2-7}
				& $i=2$ &  $s(2,4,m,0)$  &  $30$  &  $20$  &  $20$  &         \\ \cline{2-7}
				& $i=1$ &  $s(1,4,m,0)$  &  $30$  &  $20$  &  $20$  &  $-20$  \\
				
				\bottomrule
			\end{tabular}
			\caption{\red{State values $s(i,j, m ,0)$ for all $i,j \in \{1,2,3,4\}$ with $i\leq j$ and $m \in \{0,\ldots, j-i\}$.}}
			\label{tab:exState0}
		\end{center}
}	
\end{table}

\red{
At level $\ell = 1$, the subsequences are examined in a specific order as shown by the rows in Table~\ref{tab:1-exState1}: in an outer loop, index $j$ spans the range $1$ to $4$ while in an inner loop index $i$ runs from $j$ down to $1$.
In order to compute $s(i,j,m,1)$ with \eqref{equ:rec}, we first have to compute $\S(i,j,k,1)$ for all $k \in \ssq(i,j)$ according to \eqref{equ:rec_multiset}.

As $\S(i,j,k,1)$ is the union of three multi-sets, each of the three sets is given separately in the middle part of Table~\ref{tab:1-exState1}: 
On the left, indicated by $s(i+1,k,u,0)-p_i$ from \eqref{equ:rec_multiset},
the lateness values of the jobs in the subsequence, which are moved to an earlier starting time, are shown.
In the middle the lateness values $s(k+1,j,u,1)$ of the jobs in the later part of the subsequence are given. 
These values come from the previous rows in the table. 
For this reason the $s(\cdot)$ values need to be evaluated in the order stated above.
On the right the lateness value $C_k{(\ssq)} - d_i$ corresponding to the moved job is given.
}

\begin{table}[htbp]
\red{
	\begin{center}
	\renewcommand{\arraystretch}{1.2}
		\begin{tabular}{cc|c|ccc|cccc}
			\toprule
			\multicolumn{2}{c|}{} & \multirow{2}{*}{move $i\rightarrow k$} & $s(i+1,k,u,0)-p_i$    & $s(k+1,j,u,1)$     &  \multirow{2}{*}{$C_k(\ssq)-d_i$} &\multicolumn{4}{c}{$m$} \\
			 \multicolumn{2}{c|}{}  &  &  $u=0,\ldots,k-i-1$        &  $u=0,\ldots,j-k-1$    &      & $0$ & $1$ & $2$ & $3$ \\
			
			\midrule\midrule
			\multirow{2}{*}{\STAB{\rotatebox[origin=c]{90}{$j=1$}}}& \multirow{2}{*}{\STAB{\rotatebox[origin=c]{90}{$i=1$}}} 
			  & $k = 1$  & & & $-20$ & $-20$ & & & \\  \cline{3-10}
			& & $\mathbf{s(1,1,m,1)}$ & & & & $\mathbf{-20}$ & & & \\
			
            \midrule\midrule
			\multirow{5}{*}{\STAB{\rotatebox[origin=c]{90}{$j=2$}}}& \multirow{2}{*}{\STAB{\rotatebox[origin=c]{90}{$i=2$}}} 
			& $k = 2$  & & & $20$   & $20$ &  & &  \\  \cline{3-10}
			& & $\mathbf{s(2,2,m,1)}$ & & & & $\mathbf{20}$ & & &\\	\cmidrule{2-10}
			
			& \multirow{3}{*}{\STAB{\rotatebox[origin=c]{90}{$i=1$}}}  
			  & $k = 1$  &      & $20$ & $-20$  & $20$ & $-20$ & & \\
			& & $k = 2$  & $-5$ &      & $-10$  & $-5$ & $-10$ & & \\  \cline{3-10}
			& & $\mathbf{s(1,2,m,1)}$ & & & & $\mathbf{-5}$ & $\mathbf{-20}$ & & \\

			\midrule\midrule
			\multirow{9}{*}{\STAB{\rotatebox[origin=c]{90}{$j=3$}}} & \multirow{2}{*}{\STAB{\rotatebox[origin=c]{90}{$i=3$}}}
			  & $k = 3$  & & & $30$   & $30$ & & & \\  \cline{3-10}
			& & $\mathbf{s(3,3,m,1)}$ & & & & $\mathbf{30}$ & & & \\  \cmidrule{2-10}
			
			& \multirow{3}{*}{\STAB{\rotatebox[origin=c]{90}{$i=2$}}}  
			  & $k = 2$  &      &  $30$ & $20$   & $30$ & $20$ & & \\
			& & $k = 3$  & $20$ &       & $25$   & $25$ & $20$ & & \\  \cline{3-10}
			& & $\mathbf{s(2,3,m,1)}$ & & & & $\mathbf{25}$ & $\mathbf{20}$ & & \\  \cmidrule{2-10}
			
			& \multirow{4}{*}{\STAB{\rotatebox[origin=c]{90}{$i=1$}}} 
			  & $k = 1$  &           & $25$, $20$ & $-20$   & $25$ & $20$ &  $-20$ & \\ 
			& & $k = 2$  & $-5$      & $30$       & $-10$   & $30$ & $-5$ &  $-10$ & \\
			& & $k = 3$  & $5$, $-5$ &            & $-5$    &  $5$ & $-5$ &  $-5$  & \\  \cline{3-10}
			& & $\mathbf{s(1,3,m,1)}$ & & & & $\mathbf{5}$ & $\mathbf{-5}$ & $\mathbf{-20}$ & \\

			\midrule\midrule
			\multirow{14}{*}{\STAB{\rotatebox[origin=c]{90}{$j=4$}}} & \multirow{2}{*}{\STAB{\rotatebox[origin=c]{90}{$i=4$}}} 
			  & $k = 4$  &  & & $20$ & $20$ & & &  \\ \cline{3-10}
			& & $\mathbf{s(4,4,m,1)}$ & & & & $\mathbf{20}$ & & &  \\  \cmidrule{2-10}
			
			& \multirow{3}{*}{\STAB{\rotatebox[origin=c]{90}{$i=3$}}}  
			  & $k = 3$  &      & $20$ & $30$   & $30$ & $20$ & &  \\
			& & $k = 4$  & $15$ &      & $40$   & $40$ & $15$ & &  \\  \cline{3-10}
			& & $\mathbf{s(3,4,m,1)}$ & & & & $\mathbf{30}$ & $\mathbf{15}$ & & \\  \cmidrule{2-10}
			
			& \multirow{4}{*}{\STAB{\rotatebox[origin=c]{90}{$i=2$}}} 
			  & $k = 2$  &            & $30$, $15$ & $20$   & $30$ & $20$ &  $15$ & \\ 
			& & $k = 3$  & $20$       & $20$       & $25$   & $25$ & $20$ &  $20$ & \\
			& & $k = 4$  & $20$, $10$ &            & $35$   & $35$ & $20$ &  $10$ & \\  \cline{3-10}
			& & $\mathbf{s(2,4,m,1)}$ & & & & $\mathbf{25}$ & $\mathbf{20}$ & $\mathbf{10}$ & \\  \cmidrule{2-10}
			
			& \multirow{4}{*}{\STAB{\rotatebox[origin=c]{90}{$i=1$}}} 
			  & $k = 1$  &                 & $25$, $20$, $10$ & $-20$   & $25$ &  $20$ & $10$ & $-20$ \\ 
			& & $k = 2$  & $-5$            & $30$, $15$       & $-10$   & $30$ &  $15$ & $-5$ & $-10$ \\
			& & $k = 3$  & $5$, $-5$       & $20$             & $-5$    & $20$ &   $5$ & $-5$ &  $-5$ \\
			& & $k = 4$  & $5$, $-5$, $-5$ &                  & $5$     &  $5$ &   $5$ &  $-5$ & $-5$ \\  \cline{3-10}
			& & $\mathbf{s(1,4,m,1)}$ & & & & $\mathbf{5}$ & $\mathbf{5}$ & $\mathbf{-5}$ & $\mathbf{-20}$\\
			\bottomrule
		\end{tabular}
		\caption{\red{State values $s(i,j, m ,1)$ for all $i,j \in \{1,2,3,4\}$ with $i\leq j$ and $m \in \{0,\ldots, j-i\}$.}}
		\label{tab:1-exState1}
	\end{center}
}
\end{table}

\red{
The detailed description of the computation of a selected multi-set follows on the basis of Figure~\ref{fig:exMove}:
As illustrated,
when move $1 \rightarrow 2$ at level $\ell=1$ is performed in the sequence $\ssq(1,4)$, 
the values of the multi-set $\S(1,4,2,1)$ are computed according to \eqref{equ:rec_multiset} as follows:
the first multi-set contains the lateness value of job $2$ when it is moved forward by $p_1$ due to the movement of job $1$.
In this example $\bigcup_{u=0}^{0} \{s(2,2,u,0) - p_1 \} = \{-5\}$.
The second multi-set contains the lateness values of jobs $3$ and $4$ when they are optimally rearranged by moves up to level $l=1$.
Therefore, $s(3, 4, 0, 1) = 30$ (obtained by the identical move $3 \rightarrow 3$) and $s(3, 4, 1 ,1) = 15$ (obtained by move $3 \rightarrow 4$).
The third term in \eqref{equ:rec_multiset} is the lateness value of the moved job $1$ in its new position which is $C_2(\sigma_0) - d_1 = 35 - 45 = -10$.
Thus, \eqref{equ:rec_multiset} defines the multi-set $\S(1,4,2,1) = \left\{30, 15, -5, -10 \right\}$ which is also presented in Table~\ref{tab:1-exState1} together with all other multi-sets $\S(\cdot)$.
}

\begin{figure}[htbp]
\begin{center}
\begin{tikzpicture}[scale=0.15,>=triangle 45]
\node at (0,38) {\footnotesize $0$};

\draw (0,40) rectangle (25,44);
\node at (12.5,42) {$25$};
\node at (12.5,46) {$\mathbf{1}\ [45]$};
\node at (12.5,36) {$(-20)$};
\node at (25,38) {\footnotesize $25$};

\draw (25,40) rectangle (35,44);
\node at (30,42) {$10$};
\node at (30,46) {$\mathbf{2}\ [15]$};
\node at (30,36) {$(20)$};
\node at (35,38) {\footnotesize $35$};

\draw (35,40) rectangle (40,44);
\node at (37.5,42) {$5$};
\node at (37.5,46) {$\mathbf{3}\ [10]$};
\node at (37.5,36) {$(30)$};
\node at (40,38) {\footnotesize $40$};

\draw (40,40) rectangle (50,44);
\node at (45,42) {$10$};
\node at (45,46) {$\mathbf{4}\ [30]$};
\node at (45,36) {$(20)$};
\node at (50,38) {\footnotesize $50$};

\node at (0,8) {\footnotesize $0$};

\draw (0,10) rectangle (10,14);
\node at (5,12) {$10$};
\node at (5,16) {$\mathbf{2}\ [15]$};
\node at (5,6) {$(-5)$};
\node at (10,8) {\footnotesize $10$};
\draw (10,10) rectangle (35,14);
\node at (22.5,12) {$25$};
\node at (22.5,16) {$\mathbf{1}\ [45]$};
\node at (22.5,6) {$(-10)$};
\node at (35,8) {\footnotesize $35$};

\node at (35,18) {\footnotesize $35$};
\draw (35,20) rectangle (40,24);
\node at (37.5,22) {$5$};
\node at (37.5,26) {$\mathbf{3}\ [10]$};
\node at (37.5,16) {$(30)$};
\node at (40,18) {\footnotesize $40$};
\draw (40,20) rectangle (50,24);
\node at (45,22) {$10$};
\node at (45,26) {$\mathbf{4}\ [30]$};
\node at (45,16) {$(20)$};
\node at (50,18) {\footnotesize $50$};

\node at (35,-2) {\footnotesize $35$};
\draw (35,0) rectangle (45,4);
\node at (40, 2) {$10$};
\node at (40,6) {$\mathbf{4}\ [30]$};
\node at (40,-4) {$(15)$};
\node at (45, -2) {\footnotesize $45$};
\draw (45, 0) rectangle (50,4);
\node at (47.5,2) {$5$};
\node at (47.5,6) {$\mathbf{3}\ [10]$};
\node at (47.5,-4) {$(40)$};
\node at (50,-2) {\footnotesize $50$};


\draw [->,line width = 1pt] (12.5,34) -- (12.5,25) -- (22.5,25) -- (22.5,18);
\node at (17.5, 29) {\textbf{move}};
\node at (17.5, 27) {$\mathbf{1\boldsymbol\rightarrow 2}$};

\draw [decorate,decoration={brace,amplitude=5pt,mirror,raise=4ex}]
(2.5, 7) -- (7.5, 7) node[midway,yshift=-3em]{$s(2,2,0,0) - p_1$};

\draw [decorate,decoration={brace,amplitude=5pt,mirror,raise=4ex}]
(20, 7) -- (25, 7) node[midway,yshift=-3em]{$C_2(\sigma_0) - d_1$};

\draw [decorate,decoration={brace,amplitude=5pt}]
(53,24) -- (53,0) node [midway,right,xshift=.3cm,yshift=1em] {$\bigcup_{u=0}^1 \{s(3,4,u,1)\}$};
\node at (64, 11) {$ = \{ 30, 15 \}$};

\end{tikzpicture}
\caption{
\red{Illustration of computing the multi-set $\S(1,4,2,1)$ in the given example where processing times are indicated within the rectangles representing the jobs, due dates are in square brackets and lateness values are in parentheses.
}
}
\label{fig:exMove}
\end{center}
\end{figure}
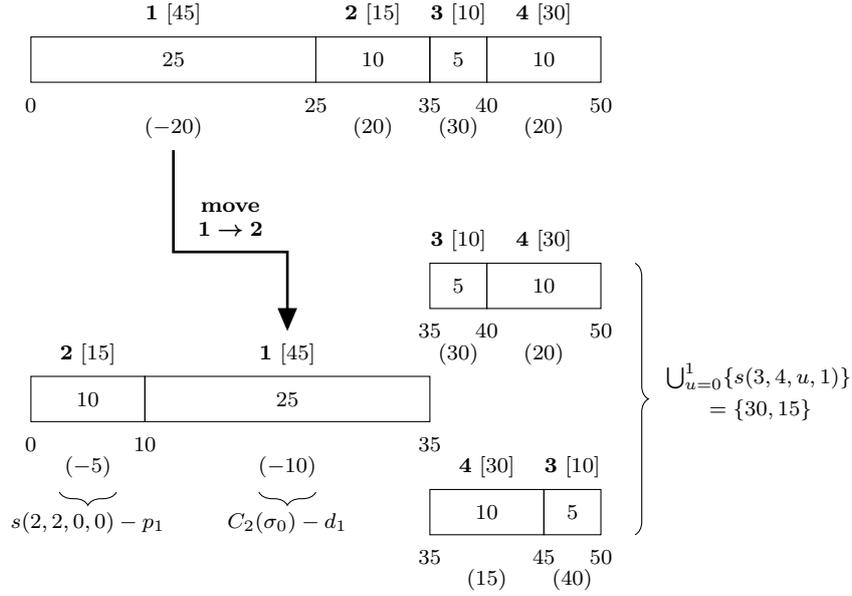

\red{
After calculating $\S(i,j,k,1)$ for all $k \in \ssq(i,j)$, the state values $s(i,j,m,1)$ can be obtained from \eqref{equ:rec}. For this purpose, the values in the multi-set $\S(i,j,k,1)$ are sorted in non-increasing order, as shown in the right part of Table~\ref{tab:1-exState1}. The value $s(i,j,m,1)$ is then obtained by talking the minimum value of all $(m+1)$-largest values from the multi-sets $\S(i,j,k,1)$. Using the table, this can be performed by calculating the minimum value of the corresponding $j-i+1$ rows in the column with the corresponding $m$-value.
For example, $s(1,4,2,1)$ corresponds to the minimum of the multi-set $\{10, -5, -5, -5\}$ which is $-5$.

}

\red{
As we can see in the last row of Table~\ref{tab:1-exState1}, the optimal solution for stack capacity $S=1$ has two late jobs since $m=2$ is the minimum value which fulfills $s(1,4,m,1)\leq 0$,
i.e.\ a schedule with two late jobs can be obtained at level $\ell=1$ starting at time $0$.
In this example multiple optimal schedules exist.
One of them, for instance, is the schedule
$\sigma = \langle 2, 1, 3, 4 \rangle$
in which jobs $3$ and $4$ are late. It can be obtained by the single move $1 \rightarrow 2$.

}

\end{ex}

\subsubsection{Algorithmic realization of the Dynamic Program}  \label{sec:algplain}

\red{In this section we illustrate a straightforward implementation of the algorithm sketched in Section~\ref{sec:dpU}.
The corresponding pseudocode 
is presented in Algorithm~\ref{alg:C3_plain}. }

The recursion requires the generation of the states from level $\ell=0$ up to $S$
to ensure that the effects of moves nested in any move of level $\ell$ have already been examined.
Furthermore, for each level, the subsequences are examined in a given order: 
in an outer loop \red{(Steps~\ref{step8}--\ref{step17})}, the index $j$ spans the range $1$ to $n$ 
while in an inner loop the index $i$ runs from $j$ down to $1$ \red{(Steps~\ref{step9}--\ref{step17})}.
This guarantees that when move $(i,k,\ell)$ 
is considered, the effects of all moves $(i',k',\ell)$ with $k<i' \leq k'$, have already been computed. 
\red{For a given triple $(i, j, \ell)$, all multi-sets $\S(i,j,k,\ell)$, $k=i,\ldots,j$, are temporarily stored in the lists $\S_k$.}
Note that \red{Steps~\ref{step12}--\ref{step14} implement the union of  the three multi-sets of Equation~\eqref{equ:rec_multiset}} and the procedure {\em Append}, used in the algorithm, simply adds an element to the end of a linked list. 

\begin{algorithm}[htbp]
	\begin{algorithmic}[1]
		\Procedure{Initialization\ }{for $\ell = 0$}	\label{step1}
		\For{$i=1,\ldots,n$}                            \label{step2}
		\For{$j=i,\ldots,n$}                            \label{step3}
		\For{$m=0,\ldots,j-i$}                          \label{step4}
		\State $s(i,j,m,0) \leftarrow$ $(m+1)$-largest value of a list with values $L_i(\ssq), \ldots, L_j(\ssq)$                                      \label{step5}
		\EndFor		
		\EndFor
		\EndFor
		\EndProcedure
		\Procedure{Recursion\ }{for $\ell \geq 1$}       \label{step6}
			\For{$\ell=1,\ldots,S$}                      \label{step7}
				\For{$j=1,\ldots,n$}                     \label{step8}
					\For{$i=j,\ldots,1$}                 \label{step9}
						\For{$k=i,\ldots,j$}    \Comment{Loop $1$}   \label{step10}
							\State Let $\S_k$ be an empty list         \label{step11}
							\For{$u = 0,\ldots, k-i-1$}               
								 {\em Append} $s(i+1,k,u,\ell-1) - p_i$ to the list $\S_k$  \label{step12}
							\EndFor
							\For{$u=0,\ldots, j-k-1$}                    
								{\em Append} $s(k+1,j,u,\ell)$ to the list $\S_k$ \label{step13}
							\EndFor
							\State {\em Append} $C_k(\ssq) - d_i$ to the list $\S_k$      \label{step14}
						\EndFor
						\For{$m=0,\ldots,j-i$}    \Comment{Loop $2$}  \label{step15}
						\State Let $\S$ be a list with the $(m+1)$-largest values of $\S_i, \ldots, \S_j$                                            \label{step16}
							\State $s(i,j,m,\ell) \leftarrow \min(\S)$   \label{step17}
						\EndFor
					\EndFor
				\EndFor
			\EndFor		
	\EndProcedure
	\State $f^{(3)}(\sigma^*) \leftarrow \min\{m: s(1,n,m,S) \leq 0\}$	 \label{step20}
    \end{algorithmic}
    \caption{Dynamic Program for minimizing the number of late jobs}
    \label{alg:C3_plain}
\end{algorithm}

The running time complexity of Algorithm~\ref{alg:C3_plain} can be analyzed as follows.
The main computation, executed in Loop $1$ (\red{Steps~\ref{step10}--\ref{step14}}) and Loop $2$ \red{(Steps~\ref{step15}--\ref{step17})} of procedure {\sc Recursion}, is performed $O(S n^2)$ times, 
namely for all levels $\ell$ and over all ordered pairs $(i,j)$.
Both, Loop $1$ and Loop $2$, \red{have total cost $O(n^2)$ since each iteration in those loops requires $O(n)$ time.}
Thus, the overall worst-case time complexity is $O(n^4 S)$ which can be bounded by $O(n^5)$.
The above discussion can be summed up in the following theorem. 
\begin{thm}\label{th:latejobs}
Algorithm~\ref{alg:C3_plain} determines an optimal solution for problem \numberlatejobs{} in $O(n^5)$ time.
\end{thm}

\subsection{Extension to the $\Omega$-constrained problem}\label{sec:omega_tardyjob}

To deal with the problem where not all jobs are movable, \red{we need to reconsider Expressions~\eqref{equ:init_late} and \eqref{equ:rec} since the multi-sets $\S(i,j,k,\ell)$ have no meaning anymore for $i \in J\setminus\Omega$. In particular, 
Recursion~\eqref{equ:rec} can be rewritten as follows:
$$
s(i,j,m,\ell) = 
\left\{
\begin{array}{ll}
\max{\!}^{(m+1)} \left( \bigcup_{u = 0}^{j-i-1} \{s(i+1,j,u,\ell)\} \cup \{C_i(\ssq)-d_i\}\right),& \mbox{ if }  i\in J\setminus\Omega, \\
\mbox{as in Expression~\eqref{equ:rec}}, & \mbox{ if }  i\in\Omega,
\end{array}\right.
\quad
$$
for all $i, j \in J$, $i\leq j$, $m = 0,\ldots,j-i$, and $\ell = 1,\ldots,S$.}
Since job $i$ cannot be moved, its lateness value is $C_i(\ssq) - d_i$, whereas the necessary reductions of the starting times for the subsequence $\ssq(i+1,j)$ are the already computed values of the states $s(i+1,j,u,\ell)$, $u=0,\ldots,j-i-1$.

\section{Minimization of the weighted number of late jobs}\label{sec:weightedlatejobs}

Here we deal with a generalization of the problem introduced in the previous Section~\ref{sect:latejobs}, and consider different job priorities. 
We therefore introduce a cost, associated to each job, which is paid if that job is late. 
This type of objective 
is studied in a number of papers (see, e.g.~\citealt{bib:b1999,bib:bpsw2005}).
We first prove that the resulting problem, namely \wnumberlatejobs, is NP-hard and then present a pseudo-polynomial 
exact solution algorithm based on dynamic programming.

\begin{thm}\label{thm:wlate}
Problem \wnumberlatejobs{} is binary NP-hard.
\end{thm}
\begin{proof}
\red{
We consider the following decision problem $\mathcal{D}$: Given an instance $\bar I$ of \wnumberlatejobs{} and an integer $Q$,
is there a feasible solution $\sigma$ of $\bar I$ such that $\sum_{j} w_j U_j(\sigma) \le Q$?
We prove that $\mathcal{D}$ is NP-complete, thus proving the hardness of \wnumberlatejobs.}

The reduction is from {\sc Equal-Cardinality Partition}: 
Given a set of positive integers $\{a_1, a_2,\ldots,a_n\}$, is there a subset $A'$ of the index set $A = \{1,\ldots, n\}$ with $\sum_{i\in A'}a_i = \sum_{i\in A\setminus A'}a_i$ and $|A'| = \frac n 2$? 
Given an instance $I$ of {\sc Equal-Cardinality Partition}, we can build an instance  \red{$\tilde I$ of $\mathcal{D}$} as follows. 
There are $n$ jobs, and for each job $i=1, \ldots, n$ we set $w_i = p_i = a_i$ and $d_i=\frac 1 2 \sum_{k=1}^n a_k$.
Moreover, the stack capacity is $S=\frac n 2$ and \red{$Q= \frac 1 2 \sum_{i=1}^n a_i$}.

It is easy to observe that \red{ $\tilde I$ is a YES-instance of $\mathcal{D}$, i.e.\ there} is a schedule $\sigma$ such that $\sum_{j} w_j U_j(\sigma) \le \red{Q}$,
if and only if $I$ is a YES-instance of {\sc Equal-Cardinality Partition}. Note that if such a schedule $\sigma$ exists, \red{ then it can be obtained by moving at most $\frac n 2$ jobs (the late ones) to the right end of the schedule.}

\end{proof}

\subsection{Dynamic Programming}\label{sec:firstdp}

\red{In this section we show that a pseudo-polynomial algorithm exists for \wnumberlatejobs, hence 
complementing the result given in Theorem~\ref{thm:wlate}.}

First of all, let us recall that all data (i.e.\ processing times, weights and due dates) are positive integer values.
Unlike the dynamic program for the unweighted case,
where for the optimal $\ell$-rearrangement of each subsequence the minimum reduction of the starting times for different number of late jobs are recorded, we now store the {\em minimum weighted number of late jobs} for different reductions of the starting time. 
%
%
\red{\begin{defi}\label{def:statew}
For each $i,j \in J$ with $i \leq j$, $t=0,\ldots,P(1,i-1)$, and $\ell = 0,\ldots,S$, let $r(i,j,t,\ell)$ 
be the {\em  minimum weighted number of late jobs} that can be achieved by rearranging the subsequence $\ssq(i,j)$
with moves up to level $\ell$, after decreasing the starting time of the subsequence by $t$.
\end{defi}
}
As an initialization we have:
\begin{equation} \label{eq:init_weightedtardy}
r(i,j,t,0) = \sum_{k \in \ssq(i,j):\, L_k(\ssq) > t} w_k \ \ \ \forall\, i,j \in J,\ i \leq j,\ \forall\, t=0,\ldots, P(1,i-1).
\end{equation}
\red{Similar to the algorithm for the unweighted case}, in the recursion, the value of state $r(i,j,t,\ell)$ is obtained by taking the minimum weighted number of late jobs over all possible moves $(i,k,\ell')$ for $k=i,\ldots,j$  at some level $\ell' \leq \ell$.
The effect of such a move is given by the sum of (at most) three terms expressing the weighted lateness of the optimal $(\ell-1)$-rearrangement of subsequence $\ssq(i+1,k)$ with a starting time decreased by $p_i$, the weighted lateness of the optimal $\ell$-rearrangement of subsequence $\ssq(k+1,j)$ and the weighted lateness contributed by the moved job $i$ which is given by
\begin{equation*}
\overline{r}(i,k,t) =
\left\{
\begin{array}{cl}
w_i, & \mbox{ if } t < C_k\left(\ssq\right) - d_i,\\
0,   & \mbox{ otherwise.}
\end{array}
\right.
\end{equation*}
Therefore, the value for the new state can be computed with the following recursion:
\red{
\begin{equation}  \label{eq:rec_weightedtardy}
r(i,j,t,\ell) = \left\{
\begin{array}{ll}
\overline{r}(i,i,t),         & \mbox{ if } i = j, \\ 
\min_{k \in \ssq(i,j)} \{\R(i,k,j,t,\ell)\}, & \mbox{ if } i < j,
\end{array}
\right.
\end{equation}
for all $i,j \in J$, $i\leq j$, for all $t = 0,\ldots, P(1,i-1)$ and for all $\ell = 1,\ldots, S$}
where
\begin{align*} 
\R(i,k,j,t,\ell) = \left\{
\begin{array}{ll}
    r(i+1,j,t,\ell) + \overline{r}(i,i,t), &   \mbox{ if } k = i, \\
    r(i+1,k,t+p_i,\ell-1) + r(k+1,j,t,\ell) + \overline{r}(i,k,t), &  \mbox{ if } i < k < j,\\
    r(i+1,j,t+p_i,\ell-1) + \overline{r}(i,j,t), &  \mbox{ if } k = j.
\end{array}
\right.
\end{align*}

The optimal solution value, $f^{(4)}(\sigma^*)$ \red{$= \sum_{j \in J} w_j U_j(\sigma^*)$}, is found by computing $r(1,n,0,S)$.

\subsubsection{Algorithmic realization of the Dynamic Program} \label{sec:weightalgplain}

\red{The pseudocode in Algorithm~\ref{alg:C4_plain} straightforwardly follows the approach described in the previous section.} 
The recursion requires, analogous to Algorithm~\ref{alg:C3_plain}, the generation of the states from level $\ell=0$ up to $S$. 
Additionally, for each level $\ell$, the subsequences $\ssq(i,j)$ are also examined in two loops where index $j$ runs from $1$ up to $n$ in an outer loop, while in an inner loop index $i$ runs from $j$ down to $1$ to ensure that certain states are already computed.

\begin{algorithm}[htbp]
	\begin{algorithmic}[1]
		\Procedure{Initialization\ }{for $\ell = 0$}	
		\For{$i=1,\ldots,n$} 
		\For{$j=i,\ldots,n$}
		\For{$t=0,\ldots,P(1,i-1)$}
		\State $r(i,j,t,0) \leftarrow$ 0
		\For{$k=i,\ldots,j$ \textbf{with} $ L_k(\ssq) > t$}
		    \State $r(i,j,t,0) \leftarrow r(i,j,t,0) + w_k$
		\EndFor
		\EndFor		
		\EndFor
		\EndFor
		\EndProcedure
		\Procedure{Recursion\ }{for $\ell \geq 1$}
			\For{$\ell=1,\ldots,S$}
				\For{$j=1,\ldots,n$}
					\For{$i=j,\ldots,1$}
					    \For{$t = 0,\ldots,P(1,i-1)$}
					        \State $r_{\min} \leftarrow \infty$
						    \For{$k=i,\ldots,j$}
							    \If{$t < C_k(\ssq) - d_i$}
	                            $\overline{r} \leftarrow w_i$
							    \Else 
	                            { $ \overline{r} \leftarrow 0$}
							    \EndIf
							    
							    \If{$i = j$}
							        $r_k \leftarrow \overline{r}$
							    \ElsIf{$k = i$}
							         $r_k \leftarrow r(i+1,j,t,\ell) + \overline{r}$
							    \ElsIf{$k = j$}
							        $r_k \leftarrow r(i+1,j,t + p_i,\ell-1) + \overline{r}$
							    \Else
							        { $r_k \leftarrow r(i+1,k,t + p_i,\ell-1) + r(k+1,j,t,\ell) + \overline{r}$}
							    \EndIf							    

                                \State $r_{\min} \leftarrow \min\{r_{\min}, r_k\}$
                            \EndFor
                            \State $r(i,j,t,\ell) \leftarrow r_{\min}$
						\EndFor
					\EndFor
				\EndFor
			\EndFor		
	\EndProcedure
	\State $f^{(4)}(\sigma^*) \leftarrow r(1,n,0,S)$	
    \end{algorithmic}
    \caption{Dynamic Program for minimizing the weighted number of late jobs}
    \label{alg:C4_plain}
\end{algorithm}

It is easy to see that the proposed algorithm has a pseudo-polynomial running time of $O(n^3 S\, P)$ where $P = P(1,n) = \sum_{j \in J} p_j$ is the sum of the processing times of all jobs.

Of course, Algorithm~\ref{alg:C4_plain} can also be used for the unweighted case (with $w_j = 1$ for all $j \in J$). 
Clearly, we then still have a pseudo-polynomial running time instead of the strongly polynomial running time stated in Theorem~\ref{th:latejobs}.

\subsection{Extension to the $\Omega$-constrained problem}  \label{sec:omega_weightedtardyjob}

To solve the $\Omega$-constrained problem, we can \red{still use Equations~\eqref{eq:init_weightedtardy} 
as an initialization but we need to adapt Recursion~\eqref{eq:rec_weightedtardy} as follows: 
$$
r(i,j,t,\ell) = 
\left\{
\begin{array}{ll}
r(i+1,j,t,\ell) + \overline{r}(i,i,t), &  \mbox{ if }  i\in J\setminus\Omega, \ i < j,\\
\mbox{as in Expression~\eqref{eq:rec_weightedtardy}}, &  \mbox{ otherwise},
\end{array}\right.
$$
for all $i, j \in J$, $i \leq j$, for all $t = 0,\ldots, P(1,i-1)$ and for all $\ell = 1,\ldots,S$.
If} job $i$ is not movable, its weighted lateness, depending on $t$, is $ \overline{r}(i,i,t)$ and the weighted lateness for the subsequence $\ssq(i+1,j)$ is simply the already computed value $r(i+1,j,t,\ell)$.

\red{ 
\subsection{Alternative Dynamic Program} \label{sec:alternative}

An alternative solution approach for \wnumberlatejobs{} can be derived by generalizing Algorithm~\ref{alg:C3_plain} to the weighted setting.} 
In this case, we define $\widetilde{s}(i,j,m,\ell)$ as a generalization of $s(i,j,m,\ell)$ where $m$ now denotes \red{the weighted number of late jobs.} Consequently, variable $m$ can take every weight value of the objective function, and therefore can be upper-bounded by the value $W = W(1,n) = \sum_{j \in J} w_j$. 
\red{
\begin{defi}\label{def:statealt}
For each $i,j \in J$ with $i \leq j$, $m = 0, \ldots, W(i,j)$ and $\ell = 0,\ldots,S$, let 
$\widetilde{s}(i,j,m,\ell)$ be the minimum decrease of the starting time of the subsequence 
$\ssq(i,j)$ that ensures a value of {\em at most} $m$ for the weighted number of late jobs 
when the subsequence can be rearranged with moves up to level $\ell$.
\end{defi}
}
The generalization of Equations~\eqref{equ:init_late}~to~\eqref{equ:rec} is as follows: the initialization for all $i,j \in J$, $i \leq j$, and for all $m=0,\ldots,W(i,j)-1$ is given by

\begin{equation*}
\widetilde{s}(i,j,m,0) = \argmin_{t \in\{ L_i(\ssq),\ldots,L_j(\ssq)\}} \left\{ \sum_{k \in \ssq(i,j):\, L_k(\ssq) > t} w_k \leq m \right\}.
\end{equation*}

%
The recursive extension rule for all $i,j \in J$, $i \leq j$, $m=0,\ldots,W(i,j)-1$ and $\ell = 1,\ldots,S$ is
\begin{align} 
\widetilde{s}(i,j,m,\ell) = \min_{k \in \ssq(i,j)} \left\{\max{\!}^{(m+1)} \widetilde{\S}(i,j,k,\ell) \right\}
\end{align}
with the modified multi-set
\begin{equation*}
\widetilde{\S}(i,j,k,\ell) = \bigcup_{u = 0}^{W(i+1,k)-1} \{\widetilde{s}(i+1,k,u,\ell-1)-p_i\} \cup \bigcup_{u = 0}^{W(k+1,j)-1} \{\widetilde{s}(k+1,j,u,\ell)\} \cup \bigcup_{u = 0}^{w_i-1}\{C_k(\ssq)-d_i\}.
\end{equation*}
Obviously, this generalization \red{of the algorithm sketched in Section~\ref{sec:dpU}} results in a worst-case computational complexity of $O(n^3 S\, W)$ which can be bounded by $O(n^4 W)$.

Summing up the above argumentation \red{ and those of Section~\ref{sec:alternative},} we obtain the following theorem. 
\begin{thm}
Problem \wnumberlatejobs{} is pseudo-polynomial solvable within
$O(n^4 B)$ running time where $B = \min\left\{\sum_{j \in J} p_j, \sum_{j \in J} w_j\right\}$.
\end{thm}

\red{
\section{Empirical analysis of the stack size effect}
\label{sec:emp}

In this section we illustrate the effect of rescheduling and in particular the influence of the stack size
for the three polynomial time algorithms presented in 
Sections~\ref{sec:weightedcompletion} to~\ref{sect:latejobs}.
We are not interested in the practical running times. 
Let us just note that the $O(n^5)$ algorithm for \numberlatejobs{}
takes considerably more time than the two $O(n^4)$ algorithms for  \totwecompletion{} and \maxlateness{}, as can be expected.
Moreover, the dynamic programming algorithms show quite small deviations in running times for instances of the same size. 

As a test bed we follow the data generation scheme by \cite{bib:potts1988}.
We choose processing times and weights independent and identically, uniformly distributed with
$p_j \sim U(1, 100)$ and
$w_j \sim U(1, 100)$.
For two parameters $d^\ell$ and $d^u$ we choose 
$d_j \sim U(P(1,n)\,d^\ell,\, P(1,n)\,d^u)$.
We consider four choices for $d^\ell$, 
namely $d^\ell \in\{0.2, 0.4, 0.6, 0.8\}$,
and all possibilities for $d^u \in \{0.2, 0.4, 0.6, 0.8, 1.0\}$ with $d^\ell \leq d^u$, which yields 14 pairs of parameter values.
For the number of jobs we take $n=50$ and $n=100$.
For each of the 28 resulting classes of instances, 20 instances were randomly generated, i.e.\ 560 instances in total.

The goal of our analysis is two-fold:
We want to analyze the potential improvement of the objective function obtained through rescheduling and we want to illustrate the utilization of the stack.
Naturally, both aspects strongly depend on the stack size.
Therefore, all our evaluations are performed for increasing stack sizes, starting from $S=1$ up to $S=25$.
Note that for practical applications in a production setting, the LIFO stack can be expected to be of moderate size with $S\leq 10$, but for the empirical analysis we consider also larger values of $S$.
It is also clear that $S$ does not depend on the number of jobs.

The results of our tests are given in Figures~\ref{fig:stack_f1} to \ref{fig:stack_f3}.
All values are taken as averages over the 280 instances for $n=50$ and $n=100$, respectively.

\begin{figure}[htbp]
    \begin{center}
       \subfigure[relative gap]{
            \label{fig:stack_f1a}
            \includegraphics[width=0.475\textwidth]{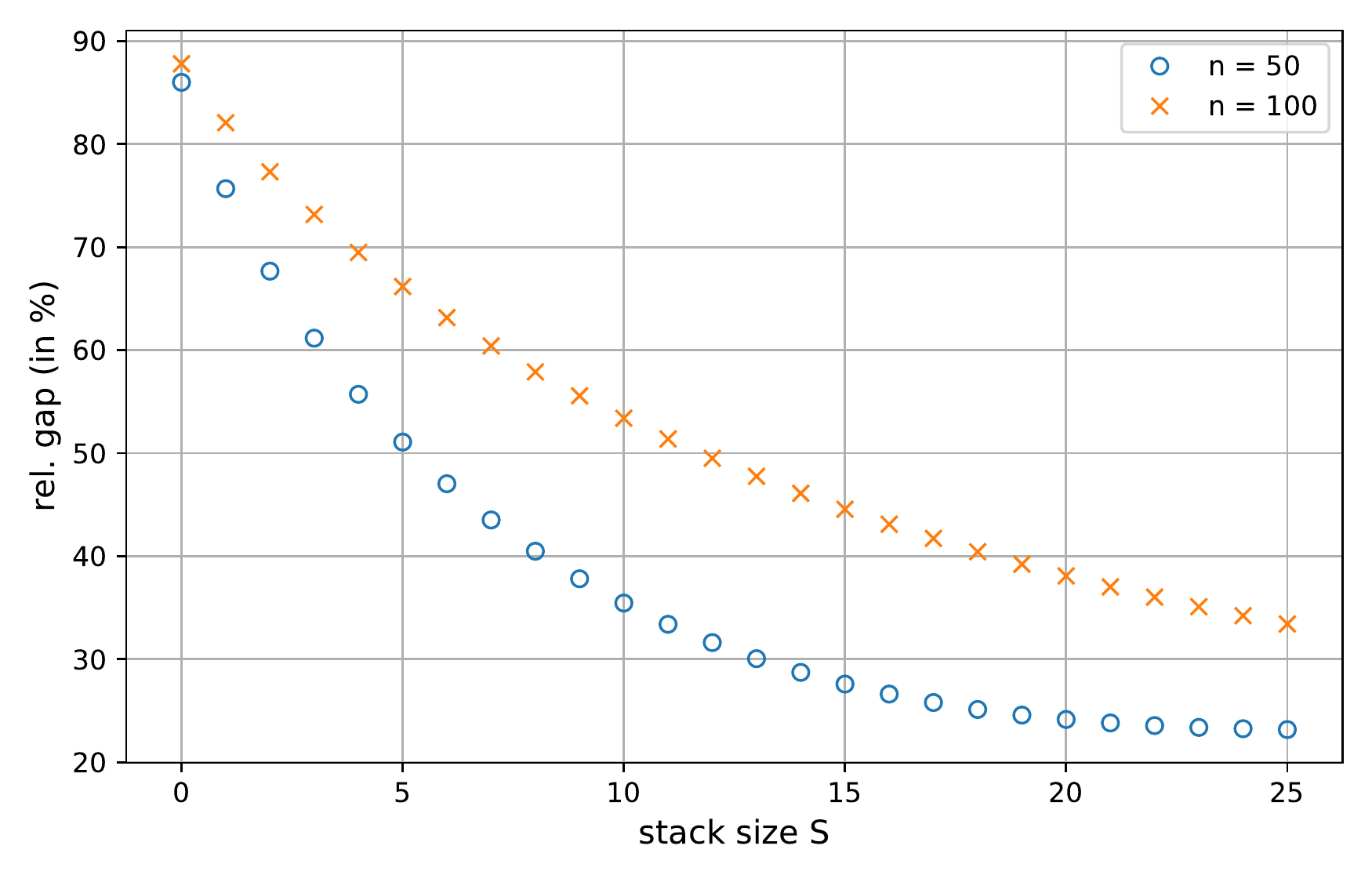}
        } 
        \subfigure[number of moves]{       
            \label{fig:stack_f1b}
            \includegraphics[width=0.475\textwidth]{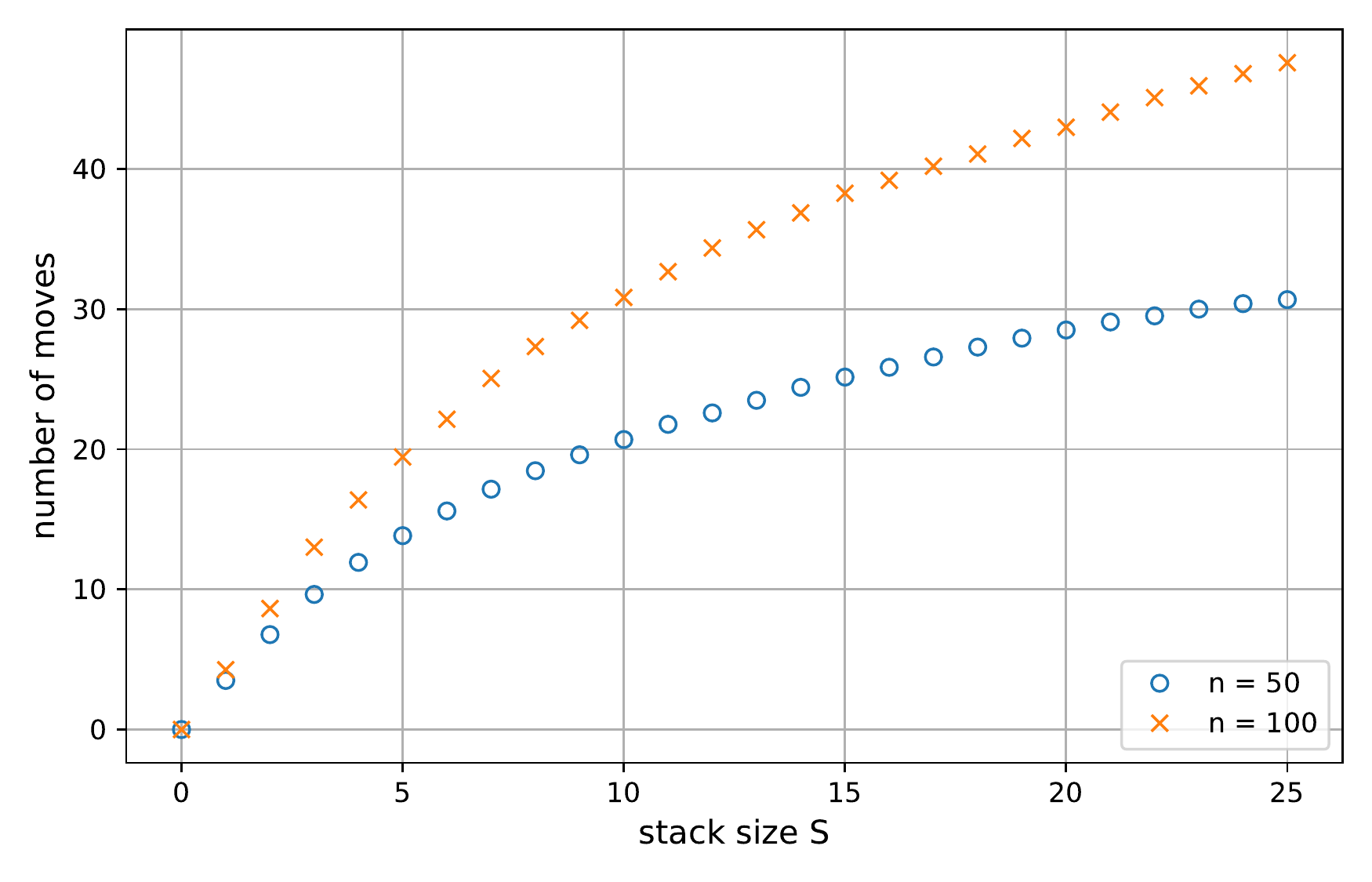}
        }
        \subfigure[maximum  stack utilization]{%
            \label{fig:stack_f1c}
            \includegraphics[width=0.475\textwidth]{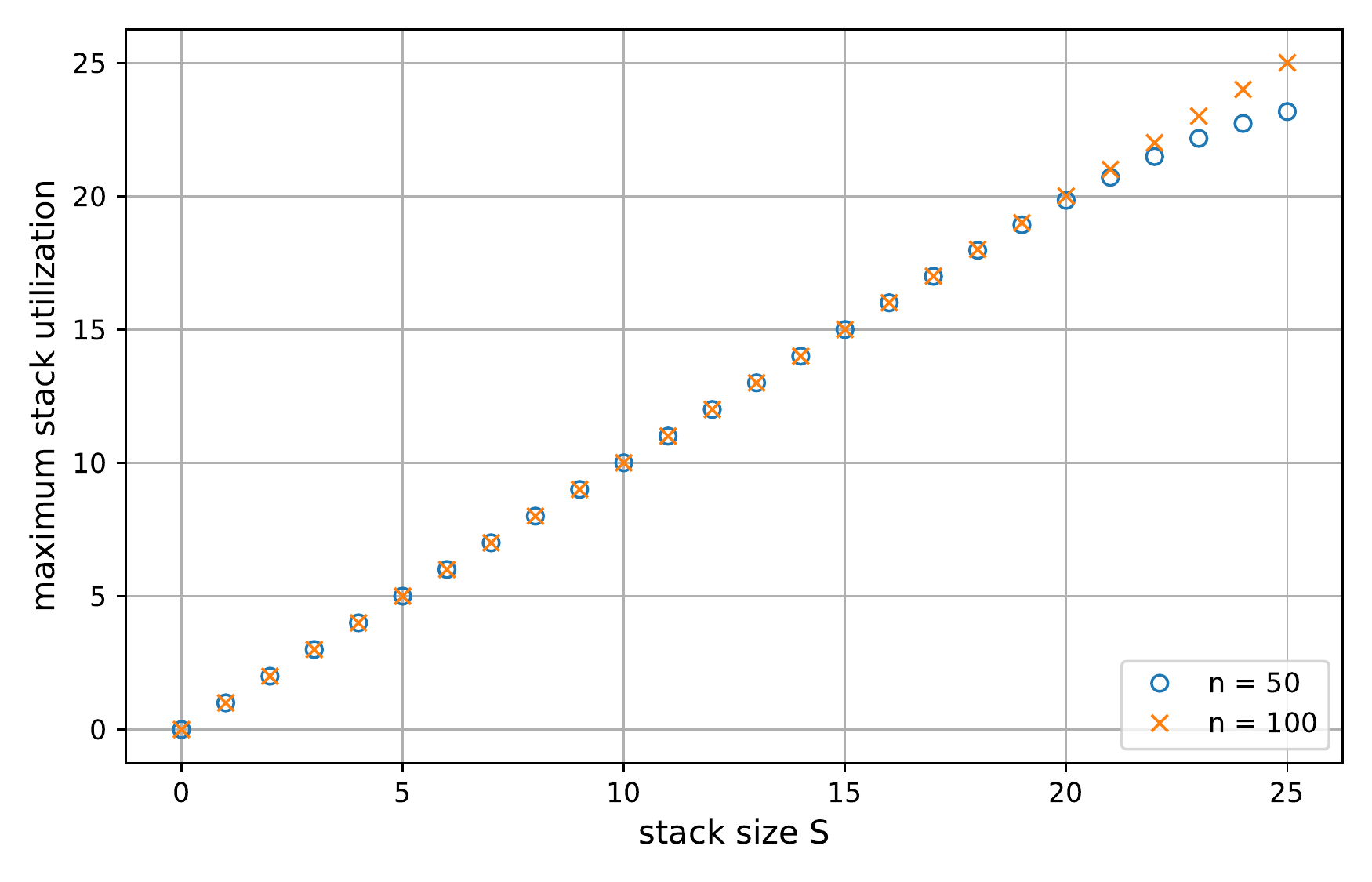}
        } 
        \subfigure[average stack utilization]{       
            \label{fig:stack_f1d}
            \includegraphics[width=0.475\textwidth]{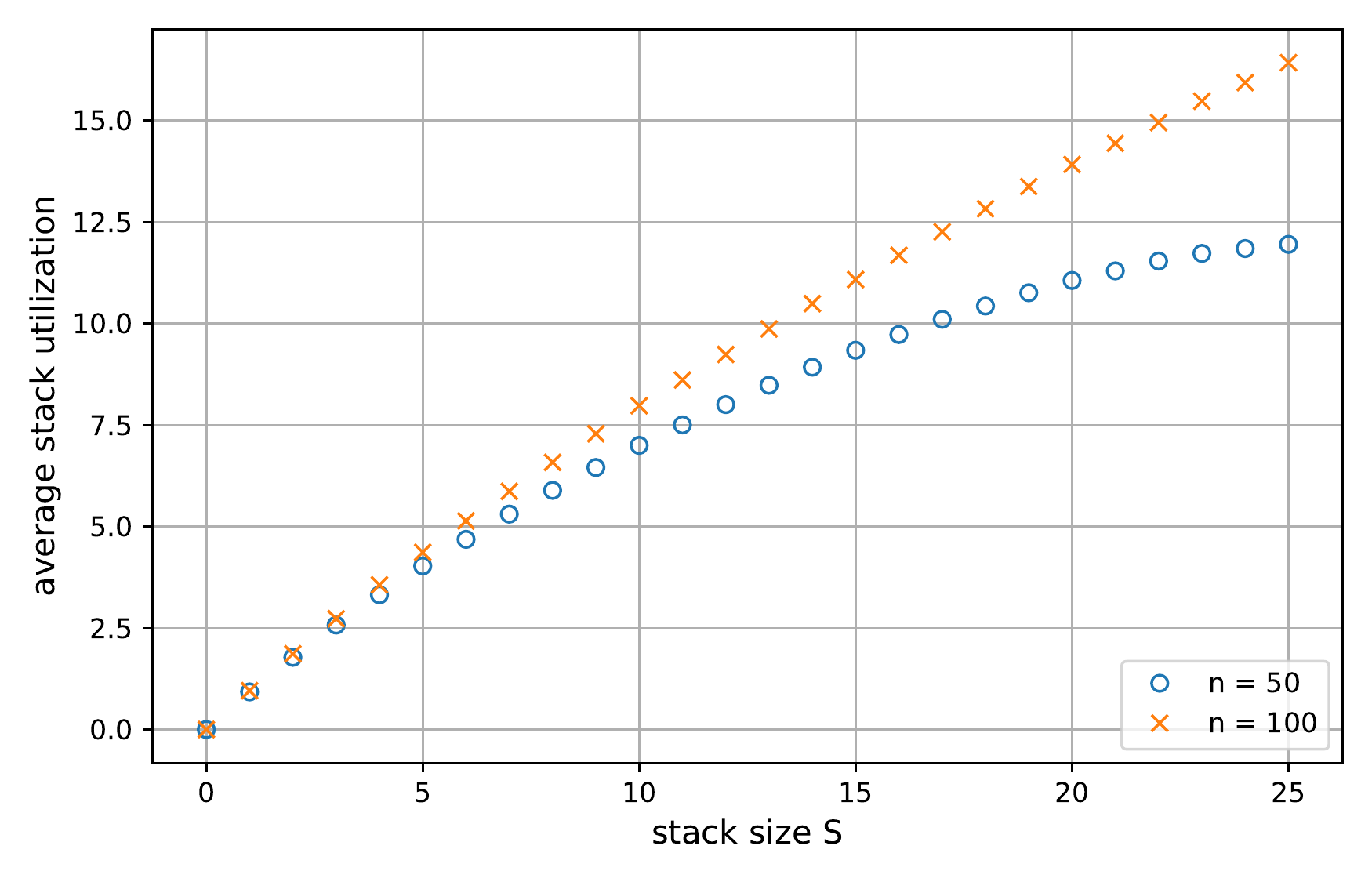}
        }
    \end{center}
    \caption{Effect of increasing stack size $S$ for \totwecompletion{}.}%
    \label{fig:stack_f1}
\end{figure}

\begin{figure}[htbp]
    \begin{center}
       \subfigure[gap relative to $C_{\max}$]{
            \label{fig:stack_f2a}
            \includegraphics[width=0.475\textwidth]{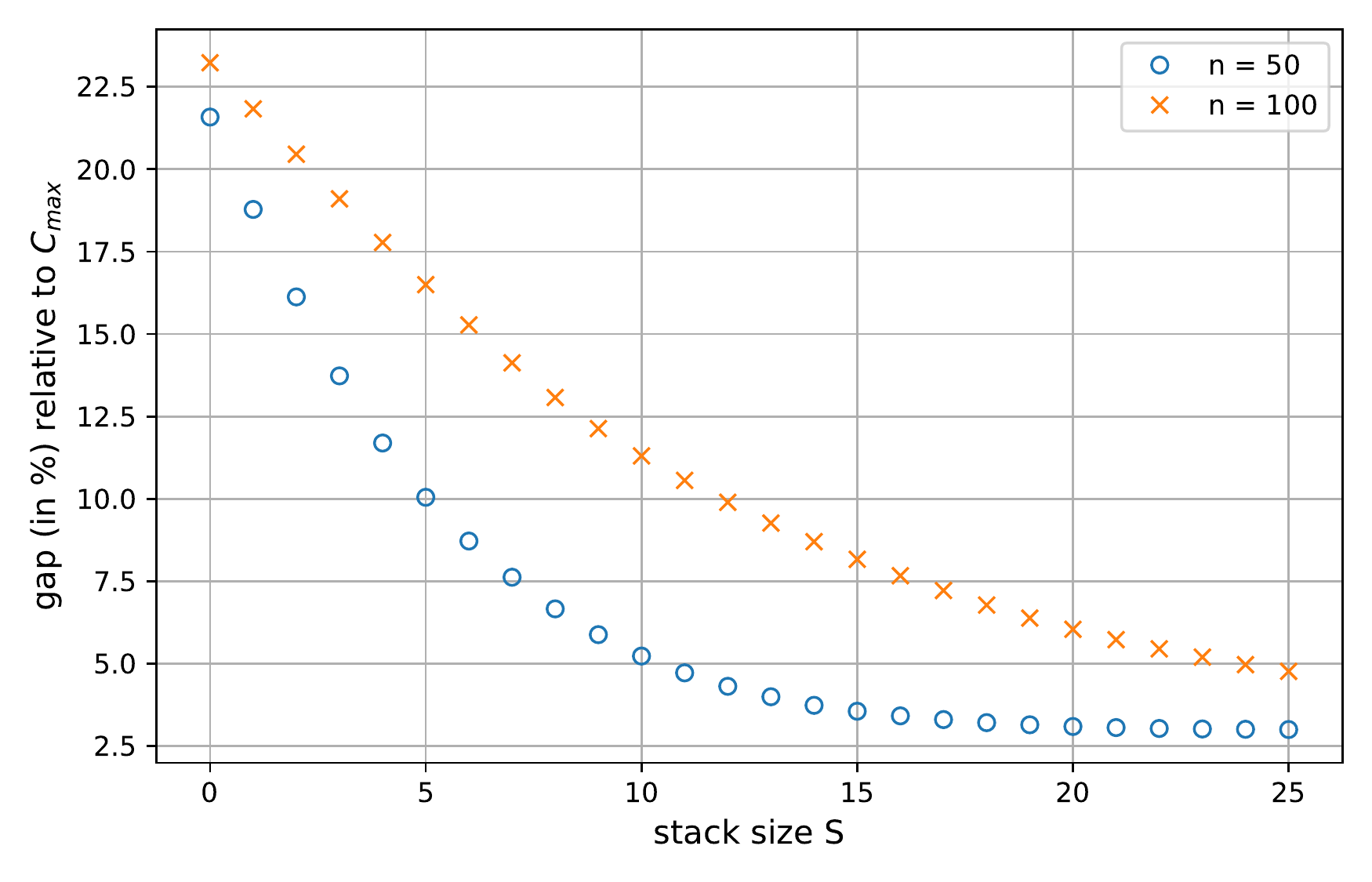}
        } 
        \subfigure[number of moves]{       
            \label{fig:stack_f2b}
            \includegraphics[width=0.475\textwidth]{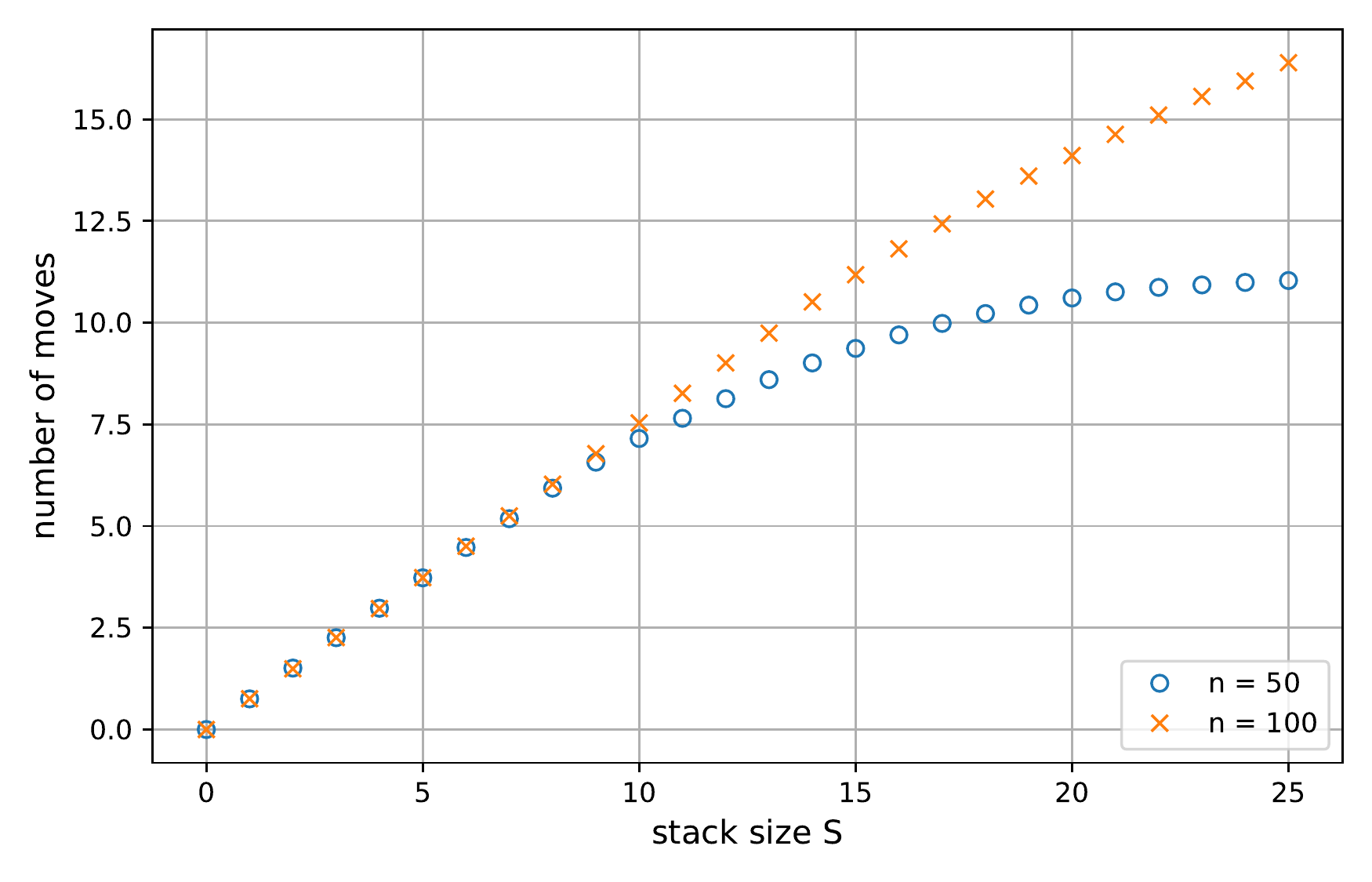}
        }
        \subfigure[maximum  stack utilization]{%
            \label{fig:stack_f2c}
            \includegraphics[width=0.475\textwidth]{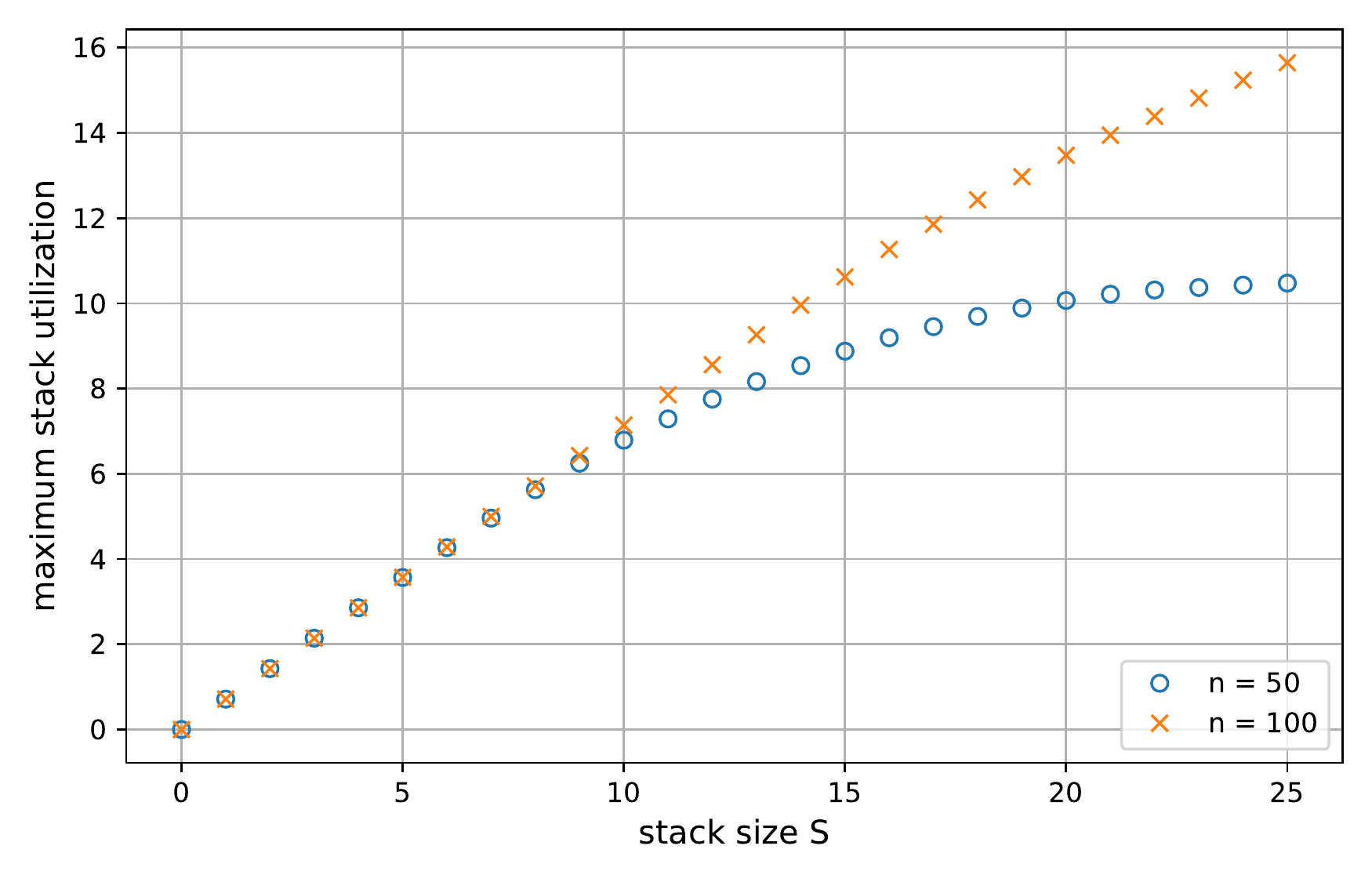}
        } 
        \subfigure[average stack utilization]{       
            \label{fig:stack_f2d}
            \includegraphics[width=0.475\textwidth]{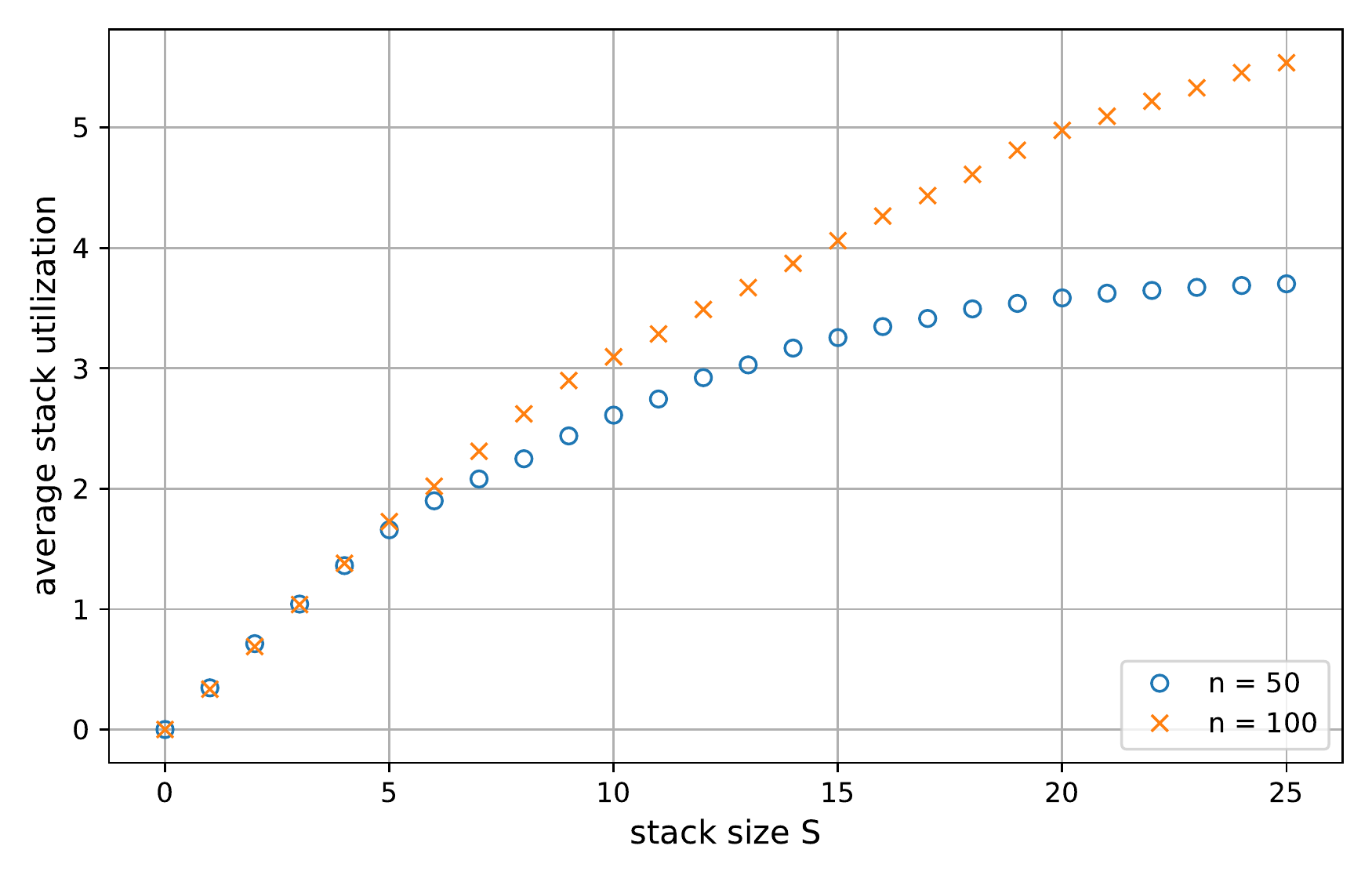}
        }
    \end{center}
    \caption{Effect of increasing stack size $S$ for \maxlateness{}.}%
    \label{fig:stack_f2}
\end{figure}

\begin{figure}[htbp]
    \begin{center}
       \subfigure[absolute gap]{
            \label{fig:stack_f3a}
            \includegraphics[width=0.475\textwidth]{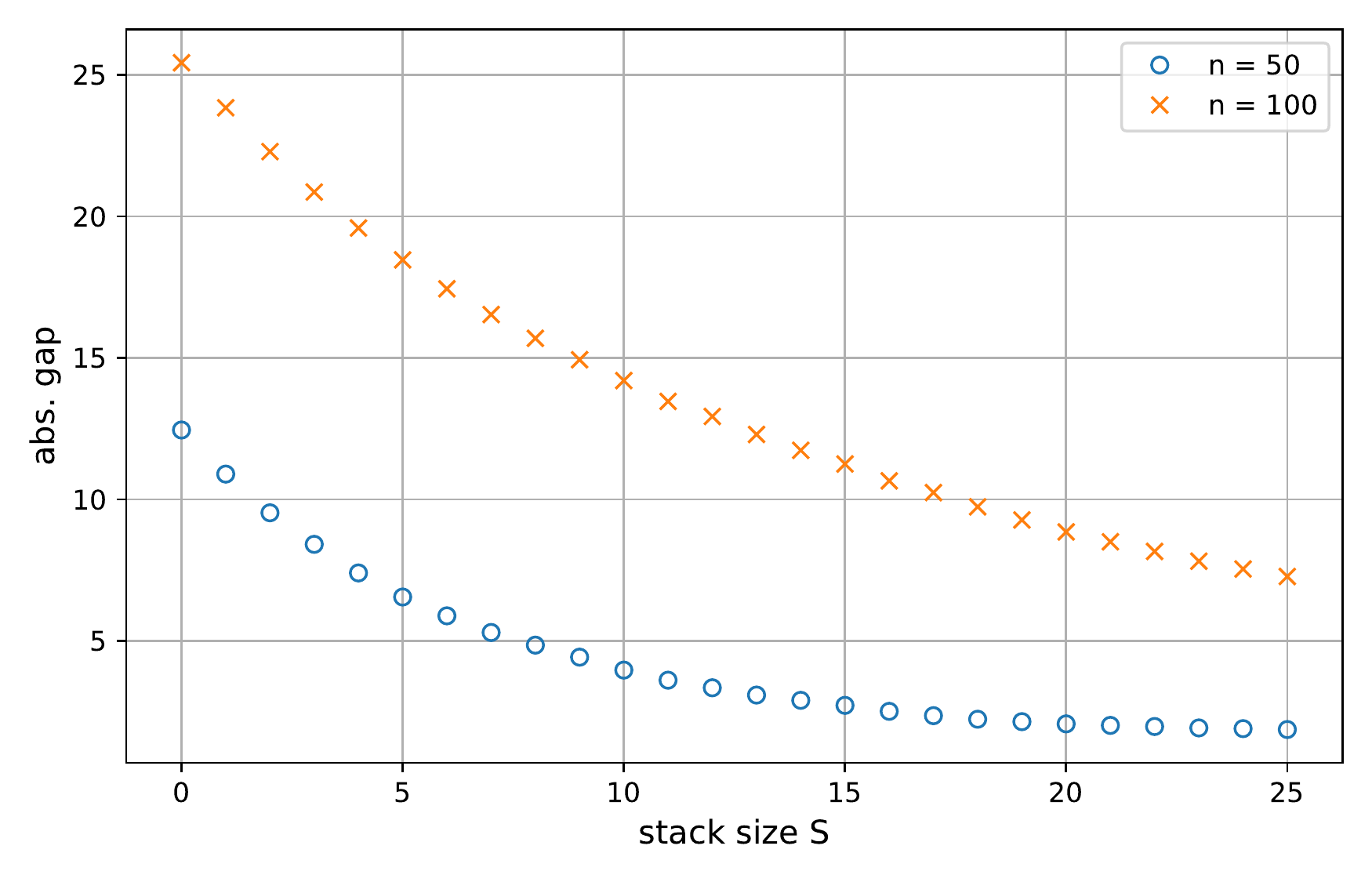}
        } 
        \subfigure[number of moves]{       
            \label{fig:stack_f3b}
            \includegraphics[width=0.475\textwidth]{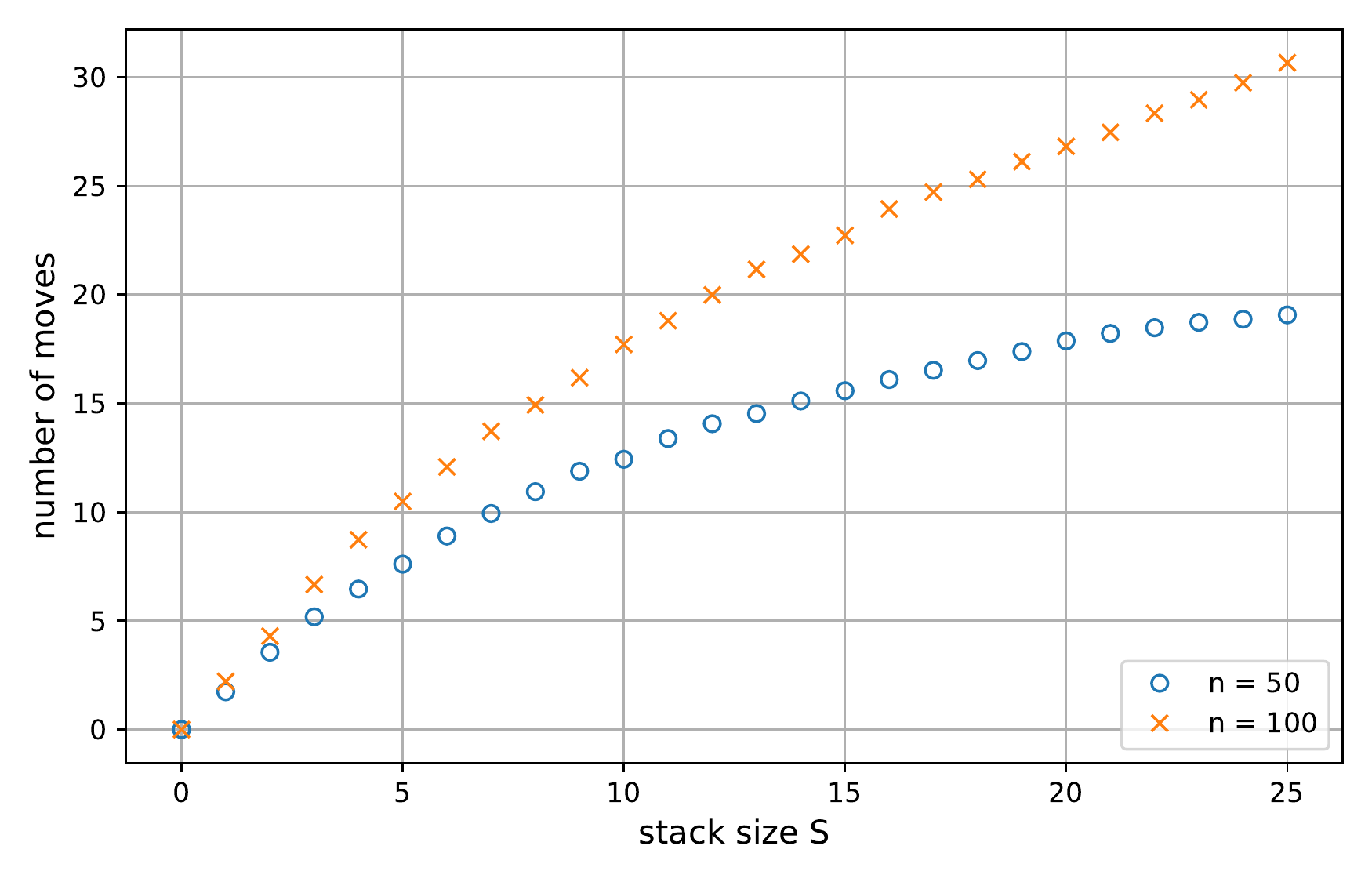}
        }
        \subfigure[maximum  stack utilization]{%
            \label{fig:stack_f3c}
            \includegraphics[width=0.475\textwidth]{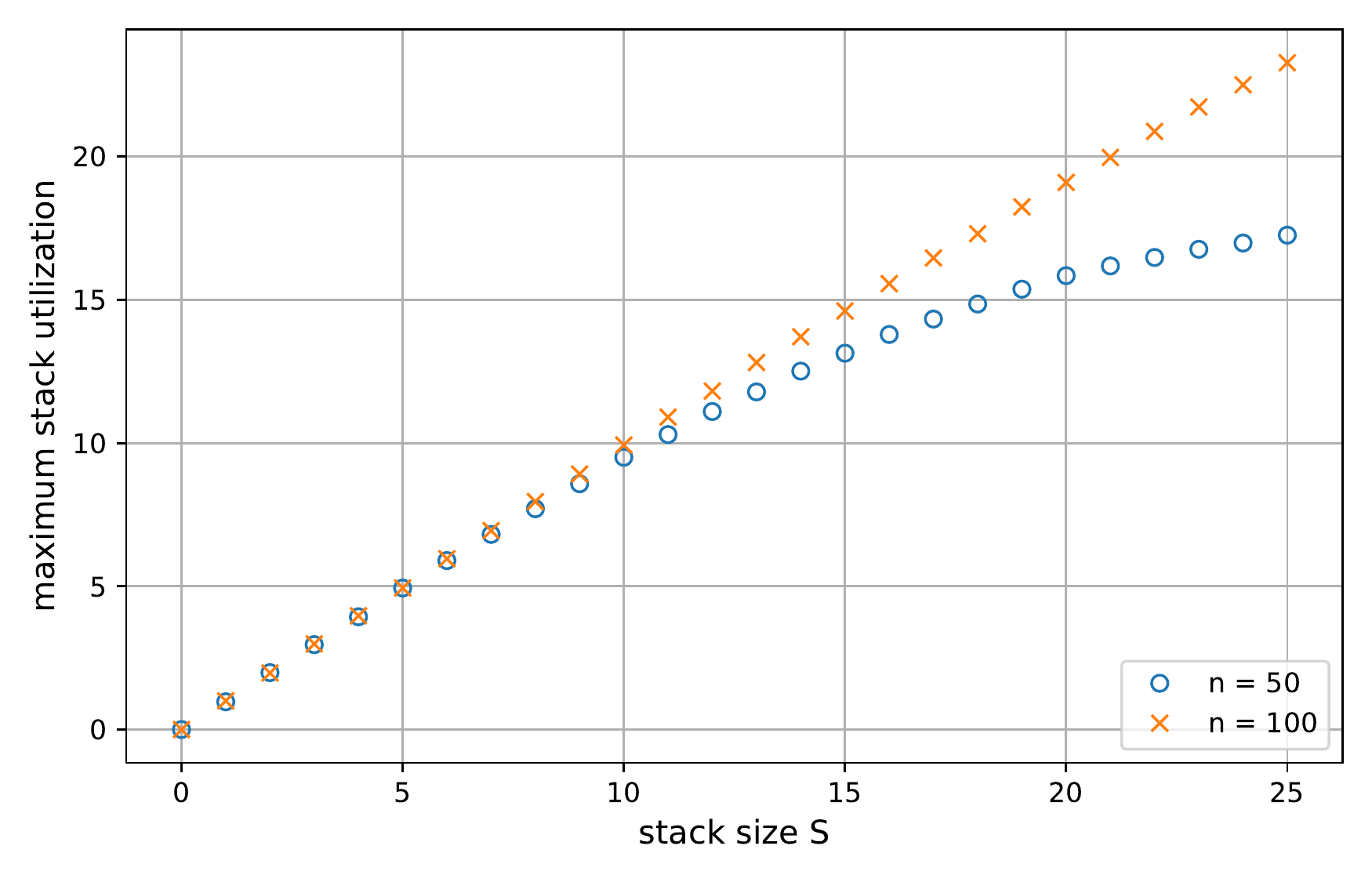}
        } 
        \subfigure[average stack utilization]{       
            \label{fig:stack_f3d}
            \includegraphics[width=0.475\textwidth]{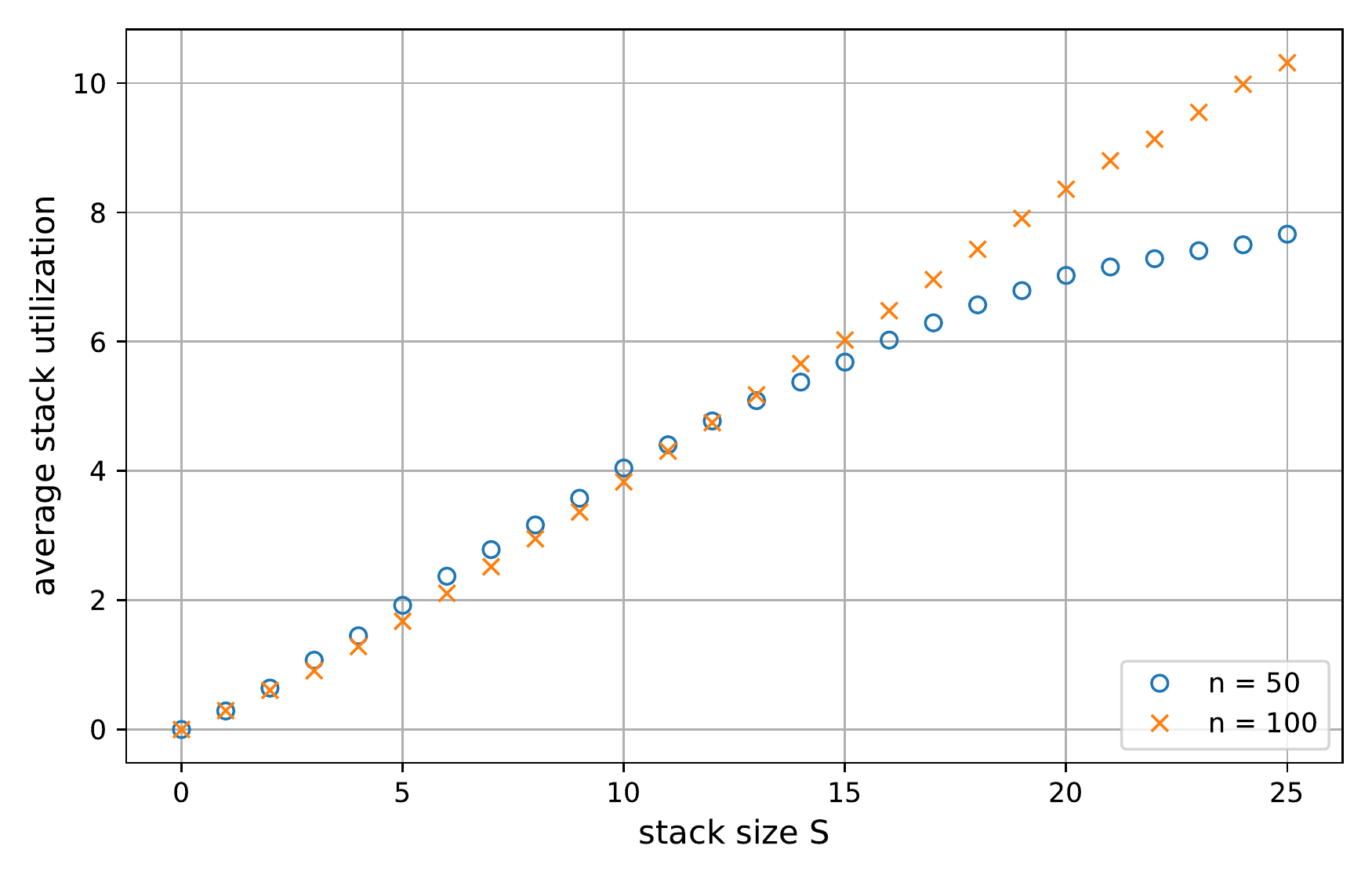}
        }
    \end{center}
    \caption{Effect of increasing stack size $S$ for \numberlatejobs{}.}%
    \label{fig:stack_f3}
\end{figure}

The two graphs in the upper row of each figure describe the effect of rescheduling on the solution.
In Figures~\ref{fig:stack_f1a}, \ref{fig:stack_f2a} and~\ref{fig:stack_f3a} we give the gap between the solution obtained from rescheduling under LIFO constraints with a certain stack size (given on the $x$-axis) and the optimal schedule (obtained by completely reordering the given jobs).
For  \totwecompletion{} the gap in Figure~\ref{fig:stack_f1a} is the difference relative to the total weighted completion time of an optimal schedule expressed in percent.
For \maxlateness{} it is not obvious how to scale the difference to the optimal schedule since the latter could have a positive, zero, or negative maximum lateness.
Thus, we take the difference between lateness after rescheduling and  lateness of an optimal schedule and divide it by $C_{\max} = P(1,n)$, see Figure~\ref{fig:stack_f2a}.
Finally, for \numberlatejobs{} the difference of the number of late jobs is given as an absolute value in Figure~\ref{fig:stack_f3a} (note that the optimal value might be zero).
It can be seen that there is a certain limit on the effect reachable by rescheduling, e.g.\ for total weighted completion time the gap hardly gets below 20\%.
This is possibly due to the LIFO constraint.

In Figures~\ref{fig:stack_f1b}, \ref{fig:stack_f2b} and~\ref{fig:stack_f3b}
we give the number of moves, i.e.\ the total number of jobs inserted into the stack during the rescheduling process.
As can be expected, a larger stack size permits much more improvement of the solution and reschedules a larger proportion of the jobs.
But from a certain stack size, the objective improves only marginally and thus also the number of moves ceases to grow.

The two graphs in the lower row of each figure describe the stack utilization. 
In Figures~\ref{fig:stack_f1c}, \ref{fig:stack_f2c} and~\ref{fig:stack_f3c}
the maximum stack utilization is given, i.e.\ the largest number of jobs contained in the stack at any time during the rescheduling process.
It turns out that for the total weighted completion time the stack size is almost always fully exploited at some point of the execution.
This is not the case when minimizing maximum lateness, especially for the smaller instances with $n=50$, where a close to optimal solution is reached with a moderate stack size, but further improvements are reachable only for a few instances as the stack size increases.
If the number of late jobs is minimized, the stack size can be mostly fully exploited for the larger instances with $n=100$, while for $n=50$ an almost optimal solution is reachable even with $S\approx 15$ (as can be seen from Figure~\ref{fig:stack_f3a})
and thus the stack is not fully utilized 
(see Figure~\ref{fig:stack_f3c}).


Finally, Figures~\ref{fig:stack_f1d}, \ref{fig:stack_f2d} and~\ref{fig:stack_f3d} show the average stack utilization.
This means that we track the number of jobs contained in the stack in each of the $n-1$ steps of the rescheduling process and take their average.
Clearly, this includes a certain phasing-in and phasing-out effect since the empty stack only starts accepting jobs at the beginning and has to be empty again until the end of the sequence.
The concave shape of the corresponding figures is due to this effect which is naturally much stronger for $n=50$ than for $n=100$.
It turns out that for maximum lateness the average stack utilization is considerably smaller than for total weighted completion time
(compare Figures~\ref{fig:stack_f1d} and \ref{fig:stack_f2d})
with number of late jobs as a close follower.
We believe that this is due to the fact that for total weighted completion time the optimal solution is unique (for distinct input values) while for the other two objectives there usually exist various different sequences with the same solution value. 
Once an optimal or very good solution value is reached, 
further rearrangements of the sequence are not beneficial any more.
}

\section{Conclusions}
\label{sec:conc}

In this paper, motivated by questions arising in manufacturing applications, 
we study the problem of rearranging a given sequence of jobs on a single machine 
in order to minimize one of four different objectives. 
\red{We rely on the standard scheduling parameters of a job, namely processing time, weight (importance) and due date.
From a practical point of view, also operating costs arising from the rearrangement steps, in particular energy consumption, could be taken explicitly into account.
}

The new sequence can be obtained respecting certain technological constraints: 
In particular, our jobs are associated to physical parts, sequenced on a conveyor that feeds a processing resource, which can be picked up by a robot. 
A \red{job} taken from the conveyor by the robot is first put into a buffer and then placed again on the conveyor in a later position of the original sequence. 
The buffer is managed as a stack with limited capacity so the last \red{job} entering the buffer is the first taken by the robot to be put down again on the conveyor. 
Due to this LIFO mechanism only a certain set of sequences can be reached starting from the initial one and this set constitutes all the feasible solutions of our problems. 

Our contribution is focused on settling the computational complexity and providing exact solution algorithms.
In particular, we are able to provide strongly polynomial (dynamic programming) solution algorithms for the minimization of 
$(i)$ total weighted completion time of the jobs, 
$(ii)$ maximum of regular functions of the completion times, in particular maximum lateness, of the jobs, and 
$(iii)$ number of late jobs. 
We also prove that $(iv)$ if we want to minimize the weighted number of late jobs, then the problem becomes (weakly) NP-hard. 
For the latter problem we present a pseudo-polynomial solution algorithm (again based on a dynamic program).

Future research should consider an empirical study concerning both the 
characterization of optimal solutions, that may be obtained starting from different input sequences, and the performance of exact solution algorithms in terms 
of numerical efficiency, with possible comparisons with alternative 
enumeration schemes.

\red{
In more general terms, it would clearly be interesting to consider other regimes for buffer management.
From a practical point of view the most relevant setting would be a queue environment implied by a FIFO mechanism.
Of course, also a random access buffer where any job can be taken from the buffer, may be well justified, although technically more complicated.
These two aspects will be subject of future research.
Moreover, it would also be interesting to consider alternative performance figures. 
We dealt with the four most common, standard objectives functions, but in the future also other choices might be relevant, in particular in the context of a real-world application \citep{bib:lvb2020}.

}




\subsection*{Funding}

Gaia Nicosia was partially supported by MIUR PRIN Project AHeAD (Efficient Algorithms for HArnessing Networked Data).
Ulrich Pferschy and Julia Resch were supported by the Field of Excellence ``COLIBRI'' at the University of Graz.
Giovanni Righini acknowledges the support of Regione Lombardia, grant agreement n. E97F17000000009, Project AD-COM.

\bibliographystyle{spbasic1}
\bibliography{references_LIFO_arXiv}

\end{document}